\documentclass{article}
\usepackage{amsfonts}
\usepackage{amssymb}
\usepackage{graphicx}
\usepackage{amsmath}
\usepackage{amsthm}
\usepackage{tikz}

\nonstopmode
\newtheorem{theorem}{Theorem}[section]
\newtheorem{proposition}[theorem]{Proposition}
\newtheorem{lemma}[theorem]{Lemma}

\newtheorem{claim}[theorem]{Claim}

\newtheorem{corollary}[theorem]{Corollary}

\theoremstyle{definition}
\newtheorem{definition}[theorem]{Definition}

\theoremstyle{remark}

\newtheorem{remark}[theorem]{Remark}

\begin{document}

\title{Model-free Superhedging Duality }
\author{Matteo Burzoni, Marco Frittelli, Marco Maggis \thanks{%
email: matteo.burzoni@unimi.it, marco.frittelli@unimi.it,
marco.maggis@unimi.it}}
\maketitle

\begin{abstract}
In a model free discrete time financial market, we prove the superhedging
duality theorem, where trading is allowed with dynamic and semi-static
strategies. We also show that the initial cost of the cheapest portfolio
that dominates a contingent claim on every possible path $\omega \in \Omega $%
, might be strictly greater than the upper bound of the no-arbitrage prices.
We therefore characterize the subset of trajectories on which this duality
gap disappears and prove that it is an analytic set.
\end{abstract}



\section{Introduction\label{Intro}}

The aim of this article is the proof of the following discrete time, model
independent version of the superhedging theorem.

\begin{theorem}[Superhedging]
\label{superH}Let $g:\Omega \mapsto \mathbb{R}$ be an $\mathcal{F}$%
-measurable random variable. Then
\begin{eqnarray*}
&&\inf \left\{ x\in \mathbb{R}\mid \exists H\in \mathcal{H}\text{ such that }%
x+(H\cdot S)_{T}\geq g\ \mathcal{M}\text{-q.s.}\right\} \\
&=&\inf \left\{ x\in \mathbb{R}\mid \exists H\in \mathcal{H}\text{ such that
}x+(H\cdot S)_{T}(\omega )\geq g(\omega )\ \forall \omega \in \Omega _{%
\mathcal{\ast }}\right\} \\
&=&\sup_{Q\in \mathcal{M}_{f}}E_{Q}[g]=\sup_{Q\in \mathcal{M}}E_{Q}[g],
\end{eqnarray*}%
where
\begin{equation}
\Omega _{\mathcal{\ast }}:=\left\{ \omega \in \Omega \mid \exists Q\in
\mathcal{M}\text{ s.t. }Q(\omega )>0\right\} .  \label{omega*}
\end{equation}
\end{theorem}

We adopt the following setting and notations: let $\Omega $ be a Polish
space and $\mathcal{F}=\mathcal{B}(\Omega )$ be the Borel sigma-algebra; $%
T\in \mathbb{N}$, $I:=\left\{ 0,...,T\right\} $, $S=(S_{t})_{t\in I}$ be an $%
\mathbb{R}^{d}$-valued stochastic process on $(\Omega ,\mathcal{F})$
representing the price process of $d\in \mathbb{N}$ assets; $\mathcal{P}$ be
the set of all probability measures on $(\Omega ,\mathcal{F})$;$\ \mathbb{F}%
^{S}:=\{\mathcal{F}_{t}^{S}\}_{t\in I}$ be the natural filtration and $%
\mathbb{F}:=\{\mathcal{F}_{t}\}_{t\in I}$ be the Universal Filtration,
namely
\begin{equation*}
\mathcal{F}_{t}:=\bigcap_{P\in \mathcal{P}}\mathcal{F}_{t}^{S}\vee \mathcal{N%
}_{t}^{P},\text{ where }\mathcal{N}_{t}^{P}=\{N\subseteq A\in \mathcal{F}%
_{t}^{S}\mid P(A)=0\}\text{;}
\end{equation*}%
$\mathcal{H}$ be the class of $\mathbb{F}$-predictable stochastic processes,
with values in $\mathbb{R}^d$, representing the family of admissible trading
strategies; $(H\cdot
S)_{T}:=\sum_{t=1}^{T}\sum_{j=1}^{d}H_{t}^{j}(S_{t}^{j}-S_{t-1}^{j})=%
\sum_{t=1}^{T}H_{t}\cdot \Delta S_{t}$ be the gain up to time $T$ from
investing in $S$ adopting the strategy $H$. We denote
\begin{eqnarray*}
\mathcal{M}:= &&\left\{ Q\in \mathcal{P}\mid S\text{ is an }\mathbb{F}\text{%
-martingale under }Q\right\} , \\
\mathcal{P}_{f}:= &&\left\{ Q\in \mathcal{P}\mid \text{supp}(Q)\text{ is
finite}\right\} , \\
\mathcal{M}_{f}:= &&\mathcal{M\cap P}_{f},
\end{eqnarray*}%
where the support of $P\in \mathcal{P}$ is defined by $\text{supp}%
(P)=\bigcap \{C\in \mathcal{F}\mid C$ closed, $P(C)=1\}$. The family of $%
\mathcal{M}$-polar sets is given by $\mathcal{N}:=\left\{ N\subseteq A\in
\mathcal{F}\ \mid \ Q(A)=0\ \forall Q\in \mathcal{M}\right\} $ and a
property is said to hold \emph{quasi surely} (q.s.) if it holds outside a
polar set. We adopt the convention $\infty -\infty =-\infty $ for those
random variables $g$ whose positive and negative part is not integrable. We
are also assuming the existence of a numeraire asset $S_{t}^{0}=1$ for all $%
t\in I.$

\paragraph{Probability free set up.}

In the statement of the superhedging theorem there is no reference to any a
priori assigned probability measure and the notions of $\mathcal{M}$, $%
\mathcal{H}$ and $\Omega _{\ast }$ only depend on the measurable space $%
(\Omega ,\mathcal{F})$ and the price process $S$. In general the class $%
\mathcal{M}$ is not dominated.

We are not imposing any restriction on $S$ so that it may describe generic
financial securities (for examples, stocks and/or options). However, in the
framework of Theorem \ref{superH} the class $\mathcal{H}$ of admissible
trading strategies requires dynamic trading in all assets. In Theorem \ref%
{superHO} below we extend this setup to the case of semi-static
trading on a finite number of options.

As illustrated in Section \ref{example}, we explicitly show that the initial
cost of the cheapest portfolio that dominates a contingent claim $g$ on
\textit{every possible path}

\begin{equation}
\inf \left\{ x\in \mathbb{R}\mid \exists H\in \mathcal{H}\text{ such that }%
x+(H\cdot S)_{T}(\omega )\geq g(\omega )\ \forall \omega \in \Omega \right\}
\label{everywhere}
\end{equation}%
can be strictly greater than $\sup_{Q\in \mathcal{M}}E_{Q}[g]$, unless some
artificial assumptions are imposed on $g$ or on the market. In order to
avoid these restrictions on the class of derivatives, it is crucial to
select the correct set of paths (i.e. $\Omega _{\ast }$) where the
superhedging strategy can be efficiently employed.

\paragraph{On the set $\Omega _{\ast }$.}

In Theorem \ref{superH}, the pathwise model independent inequality in %
\eqref{everywhere}, is replaced with an inequality involving only those $%
\omega \in \Omega $ which are weighted by at least one martingale measure $%
Q\in \mathcal{M}$. In \cite{BFM14} (see also Proposition \ref{LemmaNOpolar})
it is shown the existence of the maximal $\mathcal{M}$-polar set $N_{\ast }$%
, namely a set $N_{\ast }\in \mathcal{N}$ containing any other set $N\in
\mathcal{N}$. Moreover
\begin{equation}
\Omega _{\ast }=(N_{\ast })^{C}.  \label{33}
\end{equation}%
The inequality $x+(H\cdot S)_{T}\geq g\ $\ $\mathcal{M}$-q.s. holds by
definition outside any $\mathcal{M}$-polar set and therefore it is
equivalent, thanks to (\ref{33}), to the inequality $x+(H\cdot S)_{T}(\omega
)\geq g(\omega )\ \forall \omega \in \Omega _{\mathcal{\ast }},$ which
justifies the first equality in Theorem \ref{superH}. The set $\Omega _{\ast
}$ can be equivalently determined (see Proposition \ref{LemmaNOpolar}) via
the set $\mathcal{M}_{f}$ of martingale measures with finite support, a
property that turns out to be crucial in several proofs.

We stress that we do not make any ad hoc assumptions on the discrete time
financial model and notice that $\Omega _{\ast }$ is determined only by $S$:
indeed the set $\mathcal{M}$ can be written also as $\mathcal{M}=\left\{
Q\in \mathcal{P}\mid S\text{ is an }\mathbb{F}^{S}\text{-martingale under }%
Q\right\} $. One of the main technical result of the paper is the proof that
the set $\Omega _{\ast }$ is an \textit{analytic set} (Proposition \ref%
{propAn}) and so our findings show that the natural setup for studying this
problem is $(\Omega ,S,\mathbb{F}$,$\mathcal{H})$ with $\mathbb{F}$ the
Universal filtration (which contains the analytic sets) and $\mathcal{H}$
the class of $\mathbb{F}$-predictable processes. We also point out that we
could replace any sigma-algebra $\mathcal{F}_{t}$ with the sub sigma-algebra
generated by the analytic sets of $\mathcal{F}_{t}^{S}$.

\paragraph{On Model Independent Arbitrage and the condition $\mathcal{M}\neq
\varnothing .$}

In case $\mathcal{M}=\varnothing $ then $\Omega _{\mathcal{\ast }%
}=\varnothing $ and the theorem is trivial, as each term in the equalities
of Theorem \ref{superH} is equal to $-\infty $, provided we convene that any
$\mathcal{M}$-q.s. inequalities hold true when $\mathcal{M}=\varnothing $.
\newline
For this reason we will assume without loss of generality $\mathcal{M}\neq
\varnothing ,$ and recall that this condition can be reformulated in terms
of absence of Model Independent Arbitrages. A Model Independent $\mathcal{H}$%
-Arbitrage consists in a trading strategy $H\in \mathcal{H}$ such that $%
(H\cdot S)_{T}(\omega )>0\ \forall \omega \in \Omega $. However, as shown in
\cite{BFM14} No Model Independent $\mathcal{H}$-Arbitrage is not sufficient
to guarantees $\mathcal{M}\neq \varnothing $. Indeed we need the stronger
condition of No Model Independent $\widetilde{\mathcal{H}}$-Arbitrage to
hold, where $\widetilde{\mathcal{H}}$ is a wider class of $\widetilde{%
\mathbb{F}}$-predictable stochastic processes for a suitable enlarged
filtration $\widetilde{\mathbb{F}}$. Hence the non trivial statement in
Theorem \ref{superH} (i.e. when $\mathcal{M}\neq \varnothing $) regards the
superhedging duality under No Model Independent $\widetilde{\mathcal{H}}$%
-Arbitrage.

\subsection{Superhedging with semi-static strategies on options and stocks.}

We now allow for the possibility of static trading in a finite number of
options. Let us add to the previous market $k$ options $\Phi =(\phi
^{1},...,\phi ^{k})$ which expires at time $T$ and assume without loss of
generality that they have zero initial cost. We assume that each $\phi ^{j}$
is an $\mathcal{F}$-measurable random variable. Define $h\Phi
:=\sum_{j=1}^{k}h^{j}\phi ^{j}$, $h\in \mathbb{R}^{k}$, and
\begin{equation}
\mathcal{M}_{\Phi }:=\{Q\in \mathcal{M}_{f}\mid E_{Q}[\phi ^{j}]=0\text{ }
\forall j=1,...,k\}=\{Q\in \mathcal{M}_{f}\mid E_{Q}[h\Phi ]=0\;\forall h\in
\mathbb{R}^{k}\},  \label{Mphi}
\end{equation}
which are the options-adjusted martingale measures, and
\begin{equation}
\Omega _{\Phi }:=\left\{ \omega \in \Omega \mid \exists Q\in \mathcal{M}
_{\Phi }\text{ s.t. }Q(\omega )>0\right\} \subseteq \Omega _{\ast }.
\label{omegaphi}
\end{equation}
We have by definition that for every $Q\in \mathcal{M}_{\Phi }$ the support
satisfies $supp(Q)\subseteq \Omega _{\Phi }$. We define the superhedging
price when semi-static strategies are allowed by
\begin{equation}
\pi _{\Phi }(g):=\inf \left\{ x\in \mathbb{R}\mid \exists (H,h)\in \mathcal{%
H }\times \mathbb{R}^{k}\text{ such that }x+(H\cdot S)_{T}(\omega )+h\Phi
(\omega )\geq g(\omega )\ \forall \omega \in \Omega _{\Phi }\right\} .
\label{piPhi}
\end{equation}
With the same methodology used in the proof of Theorem \ref{superH} we will
obtain in Section \ref{secOption} the superhedging duality with semi-static
strategies, under the assumption $\mathcal{M}_{\Phi }=\{Q\in\mathcal{M}%
_f\mid supp(Q)\subseteq \Omega_{\Phi}\}$ \footnote{%
We wish to thank J. Ob\l oj and Z. Hou for pointing out that this
hypothesis is necessary for the argument used in the proof of Theorem \ref%
{superHO}. We will show in a forthcoming paper (joint with J. Ob\l
oj and Z. Hou) that the result holds in full generality dropping
this hypothesis.}:

\begin{theorem}[Super-hedging with options]
\label{superHO}Let $g:\Omega \mapsto \mathbb{R}$ and $\phi ^{j}:\Omega
\mapsto \mathbb{R}$, $j=1,...,k,$ be $\mathcal{F}$-measurable random
variables. Then
\begin{equation*}
\pi _{\Phi }(g)=\sup_{Q\in \mathcal{M}_{\Phi }}E_{Q}[g].
\end{equation*}
\end{theorem}

\subsection{Comparison with the related literature.}

In the classical case when a reference probability is fixed, this subject
was originally studied by El Karoui and Quenez \cite{ekq}; see also \cite{Ka}
and \cite{DS94} and the references cited therein.

\smallskip

In \cite{BN13} a superhedging theorem is proven in the case of a
non-dominated class of priors $\mathcal{P}^{\prime }\subseteq \mathcal{P}$.
The result strongly relies on two technical hypothesis: (i) The state space $%
\Omega $ has a product structure, $\Omega =\Omega _{1}^{T}$, where $\Omega
_{1}$ is a certain fixed Polish space and $\Omega _{1}^{t}$ is the $t$-fold
product space; (ii) The set of priors $\mathcal{P}^{\prime }$ is also
obtained as a collection of product measures $P:=P_{0}\otimes \ldots \otimes
P_{T}$ where every $P_{t}$ is a measurable selector of a certain random
class $\mathcal{P}_{t}^{\prime }\subseteq \mathcal{P}(\Omega _{1})$. $%
\mathcal{P}_{t}^{\prime }(\omega )$ represents the set of possible models
for the $t$-th period, given state $\omega $ at time $t$. An essential
requirement on $\mathcal{P}_{t}^{\prime }$ is that the graph($\mathcal{P}%
_{t}^{\prime }$) must be an analytic subset of $\Omega _{1}^{t}\times
\mathcal{P}(\Omega _{1})$. These assumptions are crucial in order to apply
the measurable selection and stochastic control arguments which lead to the
proof of the superhedging theorem. In our setting we do not impose
restrictions on the state space $\Omega $ so the result cannot be deduced
from \cite{BN13} for $\mathcal{P}^{\prime }=\mathcal{M}$. Moreover, even in
the case of $\Omega =\Omega _{1}^{T}$, the class of martingale probability
measures $\mathcal{M}$ is endogenously determined by the market and we do
not require that it satisfies any additional restrictions. Furthermore, the
techniques employed to deduce our version of the superhedging duality
theorem are completely different, as they rely on the results of \cite{BFM14}%
. Note that in the particular simple case of $\Omega:=(\mathbb{R}^{d})^T$
with $S$ the canonical process, from \cite{BFM14}, we have that $%
\Omega_{*}=\Omega$ and there are no $\mathcal{M}$-polar sets. We thus have
the equivalence between $\mathcal{P}$-q.s. and $\mathcal{M}$-q.s.
equalities. The superhedging Theorem of \cite{BN13} can be therefore applied
with $\mathcal{P}^{\prime }=\mathcal{P}$ and the two results coincide.

\smallskip

The relevance of the superhedging problem without any a priori specified set
of probability measures is revealed by the increasing amount of literature
on this topic. The problem has been studied as a particular case of a
Skorokhod Embedding Problem (see \cite{BHR01,CO11,Ho11}), following the
pioneering work \cite{Ho98} on robust hedging. The reformulation of the
superhedging duality in the framework of optimal mass transport led to
important results both in discrete and continuous time as in \cite{BHLP13,
DS13, DS14b, GHLT14, HLOST15, OH15, TT13}.

\smallskip

Different approaches are taken in \cite{AB13,Riedel}. In \cite{Riedel} the
continuity assumptions on the assets allow to embed the problem in the
linear programming framework and to obtain the desired equality in a one
period market. In \cite{AB13} from a model independent version of the
Fundamental Theorem of Asset Pricing they deduce the following superhedging
duality (Theorem 1.4 \cite{AB13})
\begin{equation}
\inf \left\{ x\in \mathbb{R}\mid \exists (H,h)\in \mathcal{H}\times \mathbb{R%
}^{k}\text{ s.t. }x+(H\cdot S)_{T}(\omega )+h\Phi (\omega )\geq g(\omega )\
\forall \omega \in \Omega \right\} =\sup_{Q\in \mathcal{M}_{\Phi }}E_{Q}[g].
\label{888}
\end{equation}%
They assume a discrete time market, with one dimensional canonical process $%
S $ on the path space $\Omega =[0,\infty )^{T}$ and an arbitrary (but non
empty) set of options on $S$ available for static trading. Theorem 1.4 in
\cite{AB13} relies on two additional technical assumptions: (i) The
existence of an option with super-linearly growing and convex payoff; (ii)
The upper semi-continuity of the claim $g$.

The example in Section \ref{example} shows that without the upper
semi-continuity of the claim $g$ the duality in (\ref{888}) fails and it
also points out that the reason for this is the insistence of superhedging
over the whole space $\Omega $, instead of over the relevant set of paths $%
\Omega _{\ast }$. Our result holds for a $d$-dimensional (not necessarily
canonical) process $S$ and does not necessitate the existence of any options.

\section{Aggregation results}

In this section we investigate when certain conditions (like superhedging or
hedging) which hold $Q$-a.s. for all $Q\in \mathcal{M}$, ensure the validity
of the correspondent pathwise conditions on $\Omega _{\ast }$.

For an arbitrary sigma-algebra $\mathcal{G}$ and for
$\mathcal{G}$-measurable random variables $X$ and $Y$, we write
$X>Y$ if $X(\omega )>Y(\omega )$ for all $\omega \in \Omega $.
When we specify $X>Y$ on a measurable set $A\subset \Omega $ it
means that $X(\omega )>Y(\omega )$ holds for all $\omega \in A$.
Similarly for $X\geq Y$ and $X=Y.$ We recall that absence of
classical arbitrage opportunities, with respect to a probability
$P\in \mathcal{P}$, is denoted by $NA(P)$. We set%
\begin{eqnarray*}
\mathcal{L}(\Omega ,\mathcal{G}):= &\{f:&\Omega \rightarrow \mathbb{R}\mid
\mathcal{G}\text{-measurable }\}, \\
\mathcal{L}(\Omega ,\mathcal{G})_{+}:= &\{f\in &\mathcal{L}(\Omega ,\mathcal{%
G})\mid f\geq 0\}.
\end{eqnarray*}
The linear space of attainable random payoffs with zero initial cost is
given by
\begin{equation*}
\mathcal{K}:=\{(H\cdot S)_{T}\in \mathcal{L}(\Omega ,\mathcal{F})\mid H\in
\mathcal{H}\}.
\end{equation*}%
Recall that the set of events supporting martingale measures $\Omega _{\ast
} $ is defined in (\ref{omega*}) and observe that the convex cones
\begin{eqnarray}
\mathcal{C}:= &\{f\in &\mathcal{L}(\Omega ,\mathcal{F})\mid f\leq k\text{ on
}\Omega _{\ast }\text{ for some }k\in \mathcal{K}\},  \label{10} \\
\mathcal{C}(Q):= &\{f\in &\mathcal{L}(\Omega ,\mathcal{F})\mid f\leq k\;Q%
\text{-a.s. for some }k\in \mathcal{K}\}.  \label{11}
\end{eqnarray}%
are related by $\mathcal{C}\subseteq \mathcal{C}(Q),$ if $Q\in \mathcal{M}.$

The main Theorem \ref{superH} relies on the following cornerstone
proposition that will be proved in Section \ref{proofProp}, as its proof
requires several technical arguments.

\begin{proposition}
\label{supsup} Let $g\in \mathcal{L}(\Omega ,\mathcal{F})$ and define
\begin{eqnarray}
\pi _{\ast }(g) &:&=\inf \left\{ x\in \mathbb{R}\mid \ \exists H\in \mathcal{%
H}\text{ s.t. }x+(H\cdot S)_{T}\geq g\quad \text{on }\Omega _{\ast }\right\}
\label{pi} \\
\pi _{Q}(g) &:&=\inf \left\{ x\in \mathbb{R}\mid \ \exists H\in \mathcal{H}%
\text{ s.t. }x+(H\cdot S)_{T}\geq g\quad Q\text{- a.s. }\right\} .
\label{piQ}
\end{eqnarray}%
Then
\begin{eqnarray}
\pi _{\ast }(g) &=&\sup_{Q\in \mathcal{M}_{f}}\pi _{Q}(g)  \label{piSup} \\
\mathcal{C} &=&\bigcap_{Q\in \mathcal{M}_{f}}\mathcal{C}(Q).
\label{intersection}
\end{eqnarray}%
In particular, if $\pi _{\ast }(g)<+\infty $ the infimum is a minimum.
\end{proposition}

\begin{corollary}
\label{aggregationSH}Let $g\in \mathcal{L}(\Omega ,\mathcal{F})$ and $x\in
\mathbb{R}$. If for every $Q\in \mathcal{M}_{f}$ there exists $H^{Q}\in
\mathcal{H}$ such that $x+(H^{Q}\cdot S)_{T}\geq g$ $Q$-a.s. then there
exists $H\in \mathcal{H}$ such that $x+(H\cdot S)_{T}(\omega )\geq g(\omega
) $ for every $\omega \in \Omega _{\ast }$.
\end{corollary}

\begin{proof}
By assumption, $g-x\in \mathcal{C}(Q)$ for every $Q\in \mathcal{M}_{f}.$
From $\mathcal{C}=\bigcap_{Q\in \mathcal{M}_{f}}\mathcal{C}(Q)$ we obtain $%
g-x\in \mathcal{C}$.
\end{proof}

\begin{corollary}[Perfect hedge]
\label{replica} Let $g\in \mathcal{L}(\Omega ,\mathcal{F})$. If for every $%
Q\in \mathcal{M}_{f}$ there exists $H^{Q}\in \mathcal{H},x^{Q}\in \mathbb{R}$
such that $x^{Q}+(H^{Q}\cdot S)_{T}=g$ $Q$-a.s. then there exists $H\in
\mathcal{H},x\in \mathbb{R}$ such that $x+(H\cdot S)_{T}(\omega )=g(\omega )$
for every $\omega \in \Omega _{\ast }$, and $x^{Q}=x$ for every $Q\in
\mathcal{M}_{f}$.
\end{corollary}

\begin{proof}
Note first that, from the hypothesis, for every $Q\in \mathcal{M}_{f}$ there
exists $H^{Q}\in \mathcal{H}$, $x^{Q}\in \mathbb{R}$ such that $%
x^{Q}+(H^{Q}\cdot S)_{T}(\omega )=g(\omega )$ for every $\omega \in supp(Q)$%
. We first show that $x^{Q}$ does not depend on $Q$. Assume there exist $%
Q_{1},Q_{2}\in \mathcal{M}_{f}$ such that $x^{Q_{1}}<x^{Q_{2}}$. For every $%
\lambda \in (0,1)$ set $Q_{\lambda }:=\lambda Q_{1}+(1-\lambda )Q_{2}\in
\mathcal{M}_{f}$. Then there exist $H^{Q_{\lambda }}\in \mathcal{H}$ and $%
x^{Q_{\lambda }}\in \mathbb{R}$ such that $x^{Q_{\lambda }}+(H^{Q_{\lambda
}}\cdot S)_{T}(\omega )=g(\omega )$ for every $\omega \in supp(Q_{\lambda
})=supp(Q_{1})\cup supp(Q_{2})$. Therefore $x^{Q_{\lambda }}+(H^{Q_{\lambda
}}\cdot S)_{T}(\omega )=g(\omega )$ for every $\omega \in supp(Q_{i}),$ for
any $i=1,2,$ and from $NA(Q_{i})$ we necessarily have that $x^{Q_{\lambda
}}=x^{i}$.\newline
Since $x+(H^{Q}\cdot S)_{T}(\omega )=g(\omega )$ for every $\omega \in
supp(Q)$ we can apply Corollary \ref{aggregationSH} which implies the
existence of $H\in \mathcal{H}$ such that $x+(H\cdot S)_{T}(\omega )\geq
g(\omega )$ on $\Omega _{\ast }$. Moreover $x-x+((H-H^{Q})\cdot
S)_{T}(\omega )\geq g(\omega )-g(\omega )$ for every $\omega \in supp(Q)$
implies $((H-H^{Q})\cdot S)_{T}(\omega )\geq 0$ for every $\omega \in
supp(Q) $. Since $NA(Q)$ holds, we conclude $((H-H^{Q})\cdot S)_{T}(\omega
)=0$ for every $\omega \in supp(Q)$. Thus for every $Q\in \mathcal{M}_{f}$
we have $x+(H\cdot S)_{T}(\omega )=g(\omega )$ on $supp(Q)$ and hence the
thesis follows from Proposition 4.18 \cite{BFM14} (or Proposition \ref%
{LemmaNOpolar}).
\end{proof}

\begin{corollary}[Bipolar representation]
\label{bipolar}Let $\mathcal{C}$ be defined in (\ref{10}). Then
\begin{equation}
\mathcal{C}=\{g\in \mathcal{L}(\Omega ,\mathcal{F})\mid E_{Q}[g]\leq
0\;\forall \,Q\in \mathcal{M}_{f}\}  \label{dualrepresentation}
\end{equation}
\end{corollary}

\begin{proof}
Clearly $\mathcal{C}\subseteq \{g\in \mathcal{L}(\Omega ,\mathcal{F})\mid
E_{R}[g]\leq 0\;\forall \,R\in \mathcal{M}_{f}\}=:\widetilde{\mathcal{C}}$.
Fix $Q\in \mathcal{M}_{f}$ \ and observe that $L^{0}(\Omega ,\mathcal{F}%
,Q)\equiv L^{1}(\Omega ,\mathcal{F},Q)\equiv L^{\infty }(\Omega ,\mathcal{F}%
,Q)$, which denote, respectively, the space of equivalent classes of $Q$%
-a.s. finite, $Q$-integrable and $Q$-a.s. bounded $\mathcal{F}$-measurable
random variables on $\Omega$. For $g\in \mathcal{L}(\Omega ,\mathcal{F})$ we
denote with the capital letter $G$ the corresponding equivalence class $G\in
L^{0}(\Omega ,\mathcal{F},Q)$. Denote also by $L_{+}^{0}(\Omega ,\mathcal{F}%
,Q)$ the $Q$-a.s. non negative elements of $L^{0}(\Omega ,\mathcal{F},Q)$.
The quotient of $\mathcal{K}$ and $\mathcal{C}(Q)$ with respect to the $Q$%
-a.s. identification $\sim _{Q}$ are denoted respectively by%
\begin{eqnarray*}
\mathcal{K}_{Q} &:&=\{K\in L^{0}(\Omega ,\mathcal{F},Q)\mid K=(H\cdot S)_{T}%
\text{ }Q-\text{a.s., }H\in \mathcal{H}\}, \\
\mathcal{C}_{Q} &:&=\{G\in L^{0}(\Omega ,\mathcal{F},Q)\mid \exists K\in
\mathcal{K}_{Q}\text{ such that }G\leq K\;Q-\text{a.s.}\}=\mathcal{K}%
_{Q}-L_{+}^{0}(\Omega ,\mathcal{F},Q).
\end{eqnarray*}%
Now we may follow the classical arguments: the convex cone $\mathcal{C}_{Q}$
is closed in probability with respect to $Q$ (see e.g. \cite{KS01} Theorem
1). As $Q\in \mathcal{M}_{f},$ $\mathcal{C}_{Q}$ is also closed in $%
L^{1}(\Omega ,\mathcal{F},Q)$ and therefore:%
\begin{equation*}
\left( \mathcal{C}_{Q}\right) ^{0}=\left\{ Z\in L^{\infty }(\Omega ,\mathcal{%
F},Q)\mid E[ZG]\leq 0\text{ }\forall G\in \mathcal{C}_{Q}\right\} \subseteq
L^{\infty }(\Omega ,\mathcal{F},Q)\cap L_{+}^{0 }(\Omega ,\mathcal{F},Q).
\end{equation*}%
Notice that $R\ll Q$ and $R\in \mathcal{M}_{f}$ if and only if $R\ll Q$ and $%
\frac{dR}{dQ}\in (\mathcal{C}_{Q})^{0}$. Hence:
\begin{eqnarray}
\left( \mathcal{C}_{Q}\right) ^{00} &=&\left\{ G\in L^{1}(\Omega ,\mathcal{F}%
,Q)\mid E[ZG]\leq 0\text{ }\forall Z\in (\mathcal{C}_{Q})^{0}\right\}  \notag
\\
&=&\left\{ G\in L^{1}(\Omega ,\mathcal{F},Q)\mid E_{R}[G]\leq 0\text{ }%
\forall R\ll Q\text{ s.t. }\frac{dR}{dQ}\in (\mathcal{C}_{Q})^{0}\right\}
\notag \\
&=&\left\{ G\in L^{1}(\Omega ,\mathcal{F},Q)\mid E_{R}[G]\leq 0\text{ }%
\forall R\ll Q\text{ s.t. }R\in \mathcal{M}_{f}\right\}  \label{13}
\end{eqnarray}%
Let $g\in \widetilde{\mathcal{C}}$. By the characterization in (\ref{13})
the corresponding $G$ belongs to $\left( \mathcal{C}_{Q}\right) ^{00}$. By
the bipolar theorem $\mathcal{C}_{Q}=\left( \mathcal{C}_{Q}\right) ^{00}$
and therefore $G\in \mathcal{C}_{Q}$ and $g\in \mathcal{C}(Q)$ (as defined
in \eqref{11}). Since this holds for any $Q\in \mathcal{M}_{f}$, from $%
\mathcal{C}=\bigcap_{Q\in \mathcal{M}_{f}}\mathcal{C}(Q)$ (Proposition \ref%
{supsup}) we conclude that $g\in \mathcal{C}$.
\end{proof}

\begin{remark}
One may ask whether the bipolar duality \eqref{dualrepresentation} implies
that $\mathcal{C}$ is closed with respect to some topology. To answer this
question let us introduce on $\mathcal{L}(\Omega ,\mathcal{F})$ the
following equivalence relation: for any $X,Y\in \mathcal{L}(\Omega ,\mathcal{%
F})$
\begin{equation*}
X\sim Y\text{ if and only if }X(\omega )-Y(\omega )=k(\omega )\text{ for
some }k\in \mathcal{K}\text{ and for every }\omega \in \Omega _{\ast }.
\end{equation*}%
Consider the quotient space $\mathbf{L}(\Omega ,\mathcal{F})=\mathcal{L}%
(\Omega ,\mathcal{F})/\sim $, denote with $[X]$ the equivalent class in $%
\mathbf{L}(\Omega ,\mathcal{F})$ having $X$ as a representative and let $V_f$
be the vector space generated by $\mathcal{M}_{f}$. We first claim that the
couple $(\mathbf{L}^{{}}(\Omega ,\mathcal{F}),V_{f})$ is a separated dual
pair under the bilinear form $\langle \cdot ,\cdot \rangle :\mathbf{L}%
^{{}}(\Omega ,\mathcal{F})\times V_{f}\rightarrow \mathbb{R}$ defined by: $%
\langle \lbrack X],\mu \rangle \mapsto E_{\mu }[X],$ for any $X\in \lbrack
X] $. Notice that the form $\langle \lbrack X],\mu \rangle \mapsto E_{\mu
}[X]$ is well posed as $E_{\mu }[k]=0$ for all $k\in \mathcal{K}$ and the
pairing is obviously bilinear. Clearly if $\mu \neq 0$ then there exists $%
\omega \in \Omega _{\ast }$ such that $\mu (\{\omega \})\neq 0$ and $E_{\mu
}[\mathbf{1}_{\omega }]\neq 0$. Thus we have showed that $\langle \lbrack
X],\mu \rangle =0,$ for every $[X], $ implies $\mu =0$. \newline
We now prove that $\langle \lbrack X],\mu \rangle =0$ for every $\mu $
implies $[X]=[0]$. By contradiction assume $[X]\neq \lbrack 0]$. By
assumption, $X$ can not be replicable at a non zero cost. Observe that if $%
X\in \lbrack X]$ is replicable at zero cost in any market $(\Omega ,\mathcal{%
F},\mathbb{F},S;Q)$ for any possible choice $Q\in \mathcal{M}_{f}$ then by
Corollary \ref{replica} $X$ is pathwise replicable for every $\omega \in
\Omega _{\ast }$, or in other words: $[X]=[0]$. \newline
Hence our assumption $[X]\neq \lbrack 0]$ implies that there exists a $Q\in
\mathcal{M}_{f}$ such that the market $(\Omega ,\mathcal{F},\mathbb{F},S;Q)$
is not complete, so that $\mathcal{M}_{e}(Q):=\{Q^{\ast }\sim Q\mid Q^{\ast
}\in \mathcal{M}\}\}\neq \{Q\}$, and $X\in \lbrack X]$ is not replicable in
such market. Then
\begin{equation*}
\inf_{Q^{\ast }\in \mathcal{M}_{e}(Q)}E_{Q^{\ast }}[X]<\sup_{Q^{\ast }\in
\mathcal{M}_{e}(Q)}E_{Q^{\ast }}[X].
\end{equation*}%
As $Q\in \mathcal{M}_{f}$ has finite support, $\mathcal{M}_{e}(Q)\subset
\mathcal{M}_{f}$ and there exists a $\mu \in \mathcal{M}_{e}(Q)\subset V_{f}$
such that $E_{\mu }[X]\neq 0,$ which is a contradiction.

Now we conclude that the cone $\mathcal{C}/_{\sim }$ is closed with respect
to the weak topology $\sigma (\mathbf{L}^{{}}(\Omega ,\mathcal{F}),V_{f})$.
Indeed, from \eqref{dualrepresentation} we obtain that%
\begin{equation*}
\mathcal{C}/_{\sim }=\{[g]\in \mathbf{L}(\Omega ,\mathcal{F})\mid
E_{Q}[g]\leq 0\;\forall \,Q\in \mathcal{M}_{f}\}=\bigcap_{Q\in \mathcal{M}%
_{f}}\{[g]\in \mathbf{L}(\Omega ,\mathcal{F})\mid E_{Q}[g]\leq 0\}
\end{equation*}%
is the intersection of $\sigma (\mathbf{L}^{{}}(\Omega ,\mathcal{F}),V_{f})$%
-closed sets.
\end{remark}

\section{Proof of Theorem \protect\ref{superH}}

%

We first recall from \cite{BFM14} the relevant properties of the set $\Omega
_{\ast }$ that will be needed several times in the proofs.

\begin{proposition}[ Proposition 4.18, \protect\cite{BFM14} ]
\label{LemmaNOpolar}In the setting described in Section \ref{Intro} we have
\begin{eqnarray}
\mathcal{M} &\neq &\varnothing \Longleftrightarrow \Omega _{\ast }\neq
\varnothing \Longleftrightarrow \mathcal{M}_{f}\neq \varnothing  \notag \\
\Omega _{\mathcal{\ast }} &=&\left\{ \omega \in \Omega \mid \exists Q\in
\mathcal{M}_{f}\text{ s.t. }Q(\omega )>0\right\} .  \label{omega*f}
\end{eqnarray}%
The complement of $\Omega _{\ast }$ is the maximal $\mathcal{M}$-polar set.
\end{proposition}

\paragraph{Proof of Theorem \protect\ref{superH}}

As already stated in the introduction, we may assume w.l.o.g. that $\mathcal{%
M}\neq \varnothing $, or equivalently $\mathcal{M}_{f}\neq \varnothing $.
The first equality of the theorem holds because of the definition of $%
\mathcal{M}$-q.s. inequality and the fact that $\Omega _{\ast }$ is the
maximal $\mathcal{M}$-polar set.

\textbf{Step 1: }Here we show that
\begin{equation*}
\inf \left\{ x\in \mathbb{R}\mid \exists H\in \mathcal{H}\text{ such that }%
x+(H\cdot S)_{T}(\omega )\geq g(\omega )\ \forall \omega \in \Omega _{%
\mathcal{\ast }}\right\} =\sup_{Q\in \mathcal{M}_{f}}E_{Q}[g].
\end{equation*}
Note first that the left hand side of the previous equation can be rewritten
as $\inf\{x\in\mathbb{R}\mid g-x\in\mathcal{C}\}$. From Corollary \ref%
{bipolar} it follows:
\begin{eqnarray*}
\inf\{x\in\mathbb{R}\mid g-x\in\mathcal{C}\}&=&\inf\{x\in\mathbb{R}\mid
E_Q[g-x]\leq 0\quad \forall Q\in\mathcal{M}_f \} \\
&=&\inf\{x\in\mathbb{R}\mid x\geq E_Q[g]\quad \forall Q\in\mathcal{M}_f \} \\
&=&\sup\{ E_Q[g]\mid Q\in\mathcal{M}_f\}.
\end{eqnarray*}

%
\textbf{Step 2: }We end the proof by showing that for any $g\in \mathcal{L}%
(\Omega ,\mathcal{F})$
\begin{equation}
\sup_{Q\in \mathcal{M}}E_{Q}[g]=\sup_{Q\in \mathcal{M}_{f}}E_{Q}[g],
\label{equalitySup}
\end{equation}%
where we adopt the convention $\infty -\infty =-\infty $ for those random
variables $g$ whose positive and negative part is not integrable. Set:
\begin{equation*}
m:=\sup_{Q\in \mathcal{M}}E_{Q}[g]\text{,\quad }l:=\sup_{Q\in \mathcal{M}%
_{f}}E_{Q}[g].
\end{equation*}%
We obviously have that $l\leq m$ so that we only have to prove the converse
inequality. If $l=\infty $ there is nothing to prove. Suppose then $l<\infty
$. We first show that
\begin{equation}
\text{if }Q\in \mathcal{M}\text{ satisfy }E_{Q}[g]>l\Rightarrow
E_{Q}[g]=\infty  \label{infinity_property}
\end{equation}%
Suppose indeed by contradiction that there exists $Q\in \mathcal{M}\setminus
\mathcal{M}_{f}$ such that $l<E_{Q}[g]<\infty $. Consider now an arbitrary
version of the process $g_{t}:=E_{Q}[g\mid \mathcal{F}_{t}]$ and extend the
original market with the asset $S_{t}^{d+1}:=g_{t}$ for $t\in I$. We
obviously have that $Q$ is a martingale measure for the extended market and
from Proposition \ref{LemmaNOpolar} this implies the existence of a finite
support martingale measure $Q_{f}$ which, by construction, belongs to $%
\mathcal{M}_{f}$. Since $E_{Q_{f}}[g]=g_{0}>l$, which is the supremum of the
expectations of $g$ over $\mathcal{M}_{f}$, we have a contradiction.

From \eqref{infinity_property} we readily infer that if $m<\infty $ then $%
l=m $. We are only left to study the case of $m=\infty $ and we show that
this is not possible under the hypothesis $l<\infty $. Consider first the
class of martingale measures $\mathcal{Q}(g)\subset \mathcal{M}$ such that $%
E_{Q}[g^{-}]=\infty $. We obviously have that $\mathcal{Q}(g)\cap \mathcal{M}%
_{f}=\varnothing $, moreover, since $l<m=\infty $ from %
\eqref{infinity_property} and from $\infty -\infty =-\infty $, there exists $%
\widetilde{Q}\in \mathcal{M}\setminus \mathcal{Q}(g)$ such that $E_{%
\widetilde{Q}}[g]=\infty $ and $E_{\widetilde{Q}}[g^{-}]<\infty $. Consider
now the sequence of claims $g_{n}:=g\wedge n$ for any $n\in \mathbb{N}$.
From $E_{\widetilde{Q}}[g^{-}]<\infty $ and Monotone Convergence Theorem we
have $E_{\widetilde{Q}}[g\wedge n]\uparrow E_{\widetilde{Q}}[g]=\infty $,
hence, there exists $\overline{n}\in \mathbb{N}$ such that $\overline{n}\geq
E_{\widetilde{Q}}[g\wedge \overline{n}]>l$. Note now that
\begin{equation}
\sup_{Q\in \mathcal{M}_{f}}E_{Q}[g\wedge \overline{n}]\leq \sup_{Q\in
\mathcal{M}_{f}}E_{Q}[g]=l<E_{\widetilde{Q}}[g\wedge \overline{n}]
\end{equation}%
Applying \eqref{infinity_property} to $g\wedge \overline{n}$ we get $E_{%
\widetilde{Q}}[g\wedge \overline{n}]=+\infty $, which is a contradiction
since the contingent claim $g\wedge \overline{n}$ is bounded.

\section{Example: forget about superhedging everywhere!}

\label{example} Let $(\Omega ,\mathcal{F})=(\mathbb{R}^{+},\mathcal{B}(%
\mathbb{R}^{+}))$. Consider a one period market ($T=1$) defined by a
non-risky asset $S_{t}^{0}\equiv 1$ for $t=0,1$ (interest rate is zero) and
a single risky asset $S_{T}^{1}(\omega )=\omega $ with initial price $%
S_{0}^{1}:=s_{0}>0$. In this market we also have two options $\Phi =(\phi
^{0},\phi ^{1})$, where $\phi ^{0}:=f^{0}(S_{T})$ is a butterfly spread
option and $\phi ^{1}:=f^{1}(S_{T})$ is a power option, i.e.
\begin{equation*}
\begin{array}{l}
f^{0}(x):=(x-K_{0})^{+}-2(x-(K_{0}+1))^{+}+(x-(K_{0}+2))^{+} \\
f^{1}(x):=(x^{2}-K_{1})^{+}. \\
\end{array}%
\end{equation*}%
Assume $K_{0}>s_{0}$, $K_{1}>(K_{0}+2)^{2}$ and that these options are
traded at prices $c_{0}=0$ and $c_{1}>0$ respectively. Set $c=(c_{0},c_{1}).$
The payoffs of these financial instruments are shown in Figure 1 for $%
K_{0}=2 $, $K_1=25$:

\begin{figure}[h]
\centering
\begin{tikzpicture}
  \draw[->] (0,0) -- (6,0) node[right] {$x$};
  \draw[->] (0,0) -- (0,6) node[above] {$y$};
  \draw[thick, scale=0.8 ,domain=0:6,smooth,variable=\x,blue] plot ({\x},{\x});
  \draw[thick, scale=0.8 ,domain=9:10,smooth,variable=\x,blue] plot ({\x},7) node[right,black]{payoff of $S^1$};
  \draw[dotted, ultra thick,scale=0.8,domain=5:5.65,smooth,variable=\x]  plot ({\x},{\x*\x-25});
  \draw[dotted, ultra thick,scale=0.8,domain=0:5,smooth,variable=\x]  plot ({\x},{0});
   \draw[dotted, ultra thick, scale=0.8 ,domain=9:10,smooth,variable=\x] plot ({\x},5.6) node[right,black]{payoff of $\phi^1$};
  \draw[thick,scale=0.8,domain=0:3,smooth,variable=\x,red]  plot ({\x},{max(\x-2,0)});
  \draw[thick,scale=0.8,domain=3:6,smooth,variable=\x,red]  plot ({\x},{max(-\x+4,0});
  \draw[thick, scale=0.8 ,domain=9:10,smooth,variable=\x,red] plot ({\x},6.3) node[right,black]{payoff of $\phi^0$};
  \draw (7,4) rectangle (10.1,6.1);

  \foreach \x in {0,1,2,3,4,5,6,7}
    \draw [scale=0.8](\x ,1pt) -- (\x ,-1pt) node[anchor=north] {$\x$};
  \foreach \y in {0,1,2,3,4,5,6,7}
    \draw [scale=0.8] (1pt,\y ) -- (-1pt,\y ) node[anchor=east] {$\y$};

\end{tikzpicture}
\caption{Payoffs.}
\end{figure}

\begin{definition}
(1) There exists a \emph{model independent arbitrage} (in the sense of
Acciaio et al. \cite{AB13}) if $\exists (H,h)\in \mathcal{H}\times \mathbb{R}%
^{2}$ such that $(H\cdot S)_{T}(\omega )+h(\Phi (\omega )-c)>0\ \forall
\omega \in \Omega $.

(2) There exists a \emph{one point arbitrage} (in the sense of \cite{BFM14})
if $\exists (H,h)\in \mathcal{H}\times \mathbb{R}^{2}$ such that $(H\cdot
S)_{T}(\omega )+h(\Phi (\omega )-c)\geq 0$ $\forall \omega \in \Omega $ and $%
(H\cdot S)_{T}(\omega )+h(\Phi (\omega )-c)>0\ $for some $\omega \in \Omega $%
.
\end{definition}

It is clear that any long position in the option $\phi ^{0}$ is a one point
arbitrage but it is not a model independent arbitrage. We have indeed that
there are No Model Independent Arbitrage as:%
\begin{equation*}
\mathcal{M}_{\Phi }\neq \varnothing .
\end{equation*}%
More precisely, any $Q\in \mathcal{M}_{\Phi }$ must satisfy $Q\left(
(K_{0},K_{0}+2)\right) =0,$ so that $(K_{0},K_{0}+2)$ is an $\mathcal{M}%
_{\Phi }$-polar set, nevertheless,
\begin{equation*}
\Omega _{\Phi }=\mathbb{R}^{+}\setminus (K_{0},K_{0}+2).
\end{equation*}%
One possible way to see this is to observe that on $\Gamma :=\mathbb{R}%
^{+}\setminus (K_{0},K_{0}+2)$ the option $\phi ^{0}$ has zero payoff and
zero initial cost so that any probability $P$, with supp$(P)\subseteq \Gamma
$, that is a martingale measure for $S^{1},\phi ^{1}$, is also a martingale
measure for $S^{0},S^{1},\phi ^{0},\phi ^{1}$. Take now $\omega _{1}=0$, $%
\omega _{2}\in (K_{0}+2,\sqrt{K_{1}})$, $\omega _{3}>\sqrt{K_{1}+c_{1}}$ and
observe that the corresponding points $x_{1}:=(-s_{0},-c_{1})$, $%
x_{2}:=(\omega _{2}-s_{0},-c_{1})$ and $x_{3}:=(\omega _{3}-s_{0},\phi
^{1}(\omega _{3})-c_{1}))$ clearly belong to $conv(\Delta X(\omega )\mid
\omega \in \Gamma )$ where $\Delta X$ is the random vector $%
[S_{1}^{1}-s_{0};\phi ^{1}-c_{1}]$. Consider now $\varepsilon :=\frac{1}{2}%
\min \{c_{1},s_{0},|\omega _{2}-s_{0}|\}$ so that for $\omega _{3}$
sufficiently large we have
\begin{equation*}
B_{\varepsilon }(0)\subseteq conv(\Delta X(\omega )\mid \omega \in \{\omega
_{1},\omega _{2},\omega _{3}\})\subseteq conv(\Delta X(\omega )\mid \omega
\in \Gamma ).
\end{equation*}%
We have therefore that $0$ is in the interior of $conv(\Delta X(\omega )\mid
\omega \in \Gamma )$ and from Corollary 4.11 item 1) in \cite{BFM14}, $%
\Omega _{\Phi }=\Gamma =\mathbb{R}^{+}\setminus (K_{0},K_{0}+2)$. Note,
moreover, that this is true for any value of the price $c_{1}>0$.

Consider now the digital options $g_{i}=F_{i}(S_{T}),$ $i=1,2$, with
\begin{eqnarray*}
F_{1}(x) &=&\mathbf{1}_{(K_{0},K_{0}+2)}(x), \\
F_{2}(x) &=&\mathbf{1}_{[K_{0},K_{0}+2]}(x)
\end{eqnarray*}%
which differ only at the extreme points of the interval $(K_{0},K_{0}+2)$
and observe that $F_{2}$ is upper semi-continuous while $F_{1}$ is not. From
the previous remark $g_{1}$ has price zero under any martingale measure $%
Q\in \mathcal{M}_{\Phi }$, so that
\begin{equation}
\sup_{Q\in \mathcal{M}_{\Phi }}E_{Q}[g_{1}]=0.  \label{555}
\end{equation}%
Define:%
\begin{equation*}
\pi _{\Omega }(g):=\inf \left\{ x\in \mathbb{R}\mid \exists (H,h)\in
\mathcal{H}\times \mathbb{R}^{2}\text{ such that }x+(H\cdot S)_{T}(\omega
)+h\Phi (\omega )\geq g(\omega )\ \forall \omega \in \Omega \right\}
\end{equation*}%
and recall that
\begin{equation*}
\pi _{\Phi }(g):=\inf \left\{ x\in \mathbb{R}\mid \exists (H,h)\in \mathcal{H%
}\times \mathbb{R}^{2}\text{ such that }x+(H\cdot S)_{T}(\omega )+h\Phi
(\omega )\geq g(\omega )\ \forall \omega \in \Omega _{\Phi }\right\}
\end{equation*}

\begin{claim}
\label{claim}In this market:

\begin{enumerate}
\item $\pi _{\Phi }(g_{1})=\sup_{Q\in \mathcal{M}_{\Phi }}E_{Q}[g_{1}]=0$%
\quad and\quad $\pi _{\Phi }(g_{2})=\sup_{Q\in \mathcal{M}_{\Phi
}}E_{Q}[g_{2}];$

\item $\pi _{\Omega }(g_{1})=\min \left\{ \frac{s_{0}}{K_{0}},1\right\}
>\sup_{Q\in \mathcal{M}_{\Phi }}E_{Q}[g_{1}]=0;$

\item $\pi _{\Omega }(g_{2})=\sup_{Q\in \mathcal{M}_{\Phi }}E_{Q}[g_{2}].$
\end{enumerate}
\end{claim}

\begin{remark}
(i) Item (1) is in agreement with the conclusion of Theorem \ref{superHO}.

(ii) Item (2) shows instead that the superhedging duality with respect to
the whole $\Omega $ does not hold for the claim $g_{1}$ (which is even
bounded). Note that in this example all the hypothesis of Theorem 1.4 in
\cite{AB13} are satisfied except for the upper semi-continuity of $g_{1}$.
\end{remark}

As the comparison between $g_{1}$ and $g_{2}$ in items (2) and (3) shows,
the assumption of upper semi-continuity of the claim seems artificial from
the financial point of view, even though necessary for the validity of
Theorem 1.4 in \cite{AB13}.

Our results demonstrates that it is possible to obtain a superhedging
duality on the relevant set $\Omega _{\Phi }$ (or $\Omega _{\ast }$ when
there are no options) for \textit{any measurable claim}, regardless of the
continuity assumptions (as well as without the existence of an option with
super-linear payoff).

\begin{proof}[Proof of the Claim \protect\ref{claim}]
Item (1) holds thanks to Theorem \ref{superH} since in the
one-period model there is no difference between dynamic and static
hedging. Notice also that the equalities $\pi _{\Phi
}(g_{1})=0=\sup_{Q\in \mathcal{M}_{\Phi
}}E_{Q}[g_{1}] $ are consequences of (\ref{555}) and the fact that $%
(H,h)=(0,0)$ is a superhedging strategy for $g_{1}$ on $\Omega _{\Phi }$. As
$g_{2}$ is upper semi-continuous, the superhedging duality in item (3) holds
thanks to Theorem 1.4 in \cite{AB13}, see (\ref{888}). In the remaining of
this section we conclude the proof by showing $\pi _{\Omega }(g_{1})=\min
\left\{ \frac{s_{0}}{K_{0}},1\right\}= \frac{s_{0}}{K_{0}}$ (by the
assumption $K_0>s_0$) and hence item (2).

Let us consider the model independent superhedging strategies i.e. the set
of $(H,h)\in \mathbb{R}^{2}\times \mathbb{R}^{2}$ such that $x+(H\cdot
S)_{T}(\omega )+h\Phi (\omega )\geq g_{1}(\omega )$ for any $\omega \in
\Omega $. Any admissible trading strategy is given by $%
(H,h):=[H^{0},H^{1},h^{0},h^{1}]\in \mathbb{R}^{4}$ which correspond to
positions in the securities $[S^{0},S^{1},\phi ^{0},\phi ^{1}]$ so that
\begin{equation}
\begin{array}{rll}
\text{price:} & V_{0}(H,h):= & H^{0}+H^{1}s_{0}+h^{1}c_{1} \\
\text{payoff:} & V_{T}(H,h):= & H^{0}+H^{1}\omega +h^{0}\phi ^{0}(\omega
)+h^{1}\phi ^{1}(\omega )%
\end{array}
\label{payoff_price}
\end{equation}

\textbf{Trivial super-hedges} There are two immediate strategies whose
terminal payoff is a super-hedge for $g_1$.

\begin{enumerate}
\item $S^{0}$ (i.e. $H^{0}=1$ in (\ref{payoff_price}) and $%
H^{1}=h^{0}=h^{1}=0$) with initial cost $1$.

\item $\frac{1}{K_{0}}S^{1}$ (i.e. $H^{1}=\frac{1}{K_{0}}$ in (\ref%
{payoff_price}) and $H^{0}=h^{0}=h^{1}=0$) with initial cost $\frac{s_{0}}{%
K_{0}}$.
\end{enumerate}

\noindent Consider now a generic superhedging strategy $(H,h)$ for the
option $g_{1}$ and suppose first that $H^{1}\geq 0$.

Observe that for every $\omega \in \lbrack 0,K_{0}]$ we have: $%
V_{T}(H,h)(\omega )=H^{0}+H^{1}\omega $ and $g_{1}(\omega )=0$. If $H^{0}<0$
there exists $\widetilde{\omega }\in \lbrack 0,K_{0}]$ such that $H^{0}+H^{1}%
\widetilde{\omega }<0=g_{1}(\widetilde{\omega })$ so that the strategy does
not dominate the payoff of $g_{1}$. Necessarily $H^{0}\geq 0$.

\begin{description}
\item[$h^{1}\neq 0$ is not optimal for super-hedging $g_{1}$] If $h^{1}\neq
0 $ we necessarily have $h^{1}\geq 0$, otherwise $V_{T}(H,h)(\omega )<0$ for
$\omega $ large enough (because of the super-linearity of $f^{1}$) and $(H,h)
$ is not a super-hedge for $g_{1}$. Since $f^{1}(x)=0$ on $(K_{0},K_{0}+2)$
and $c_{1}>0$, the most convenient super-hedge is with $h^{1}=0$ (cfr Figure
2).

\item
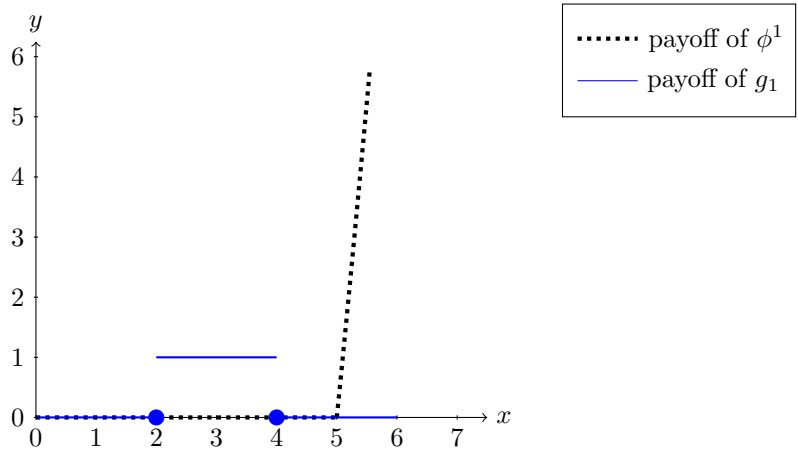
\begin{figure}[h]
\centering
\begin{tikzpicture}
  \draw[->] (0,0) -- (6,0) node[right] {$x$};
  \draw[->] (0,0) -- (0,5) node[above] {$y$};

   \draw[dotted,ultra thick,scale=0.8,domain=5:5.55,smooth,variable=\x]  plot ({\x},{\x*\x-25});
  \draw[dotted,ultra thick,scale=0.8,domain=0:5,smooth,variable=\x]  plot ({\x},{0});
   \draw[dotted,ultra thick, scale=0.8 ,domain=9:10,smooth,variable=\x] plot ({\x},6.3) node[right]{payoff of $\phi^1$};
  \draw[thick,scale=0.8,domain=0:2,smooth,variable=\x,blue]  plot ({\x},{0});
  \draw[thick,scale=0.8,domain=4:6,smooth,variable=\x,blue]  plot ({\x},{0});
  \draw[thick,scale=0.8,domain=2:4,smooth,variable=\x,blue]  plot ({\x},{1});

  \draw[scale=0.8 ,domain=9:10,smooth,variable=\x,blue] plot ({\x},5.6) node[right,black]{payoff of $g_1$};
  \draw (7,4) rectangle (10.1,5.5);

  \foreach \x in {0,1,2,3,4,5,6,7}
    \draw [scale=0.8](\x ,1pt) -- (\x ,-1pt) node[anchor=north] {$\x$};
  \foreach \y in {0,1,2,3,4,5,6}
    \draw [scale=0.8] (1pt,\y ) -- (-1pt,\y ) node[anchor=east] {$\y$};

    \draw [fill,blue] (1.6,0) circle [radius=0.1];
    \draw [fill,blue] (3.2,0) circle [radius=0.1];

\end{tikzpicture}
\caption{$\protect\phi ^{1}$ has no positive wealth on $(K_{0},K_{0}+2)$.}
\end{figure}

From now on with no loss of generality $h^{1}=0$.

\item[$h^{0}\neq 0$ is not optimal for super-hedging $g_{1}$] Since $\phi
^{0}$ has a positive payoff, if $h^{0}\neq 0$ we might take $h^{0}\geq 0$
otherwise we have a better super-hedge (at the same cost) by replacing $%
h^{0}\phi ^{0}$ with the zero portfolio. Suppose now $h^{0}>0$. By recalling
that $H^{0},H^{1}\geq 0$ we note that $V_{T}(H,h)$ as in (\ref{payoff_price}%
) satisfies
\begin{equation*}
\inf_{\omega \in (K_{0},K_{0}+2)}H^{0}+H^{1}\omega +h^{0}\phi ^{0}(\omega
)=H^{0}+H^{1}K_{0}
\end{equation*}%
so that the same super-hedge is achieved by trading only in $S^{0}$ and $%
S^{1}$. In other words with no loss of generality $h^{0}=0$ (cfr Figure 3)

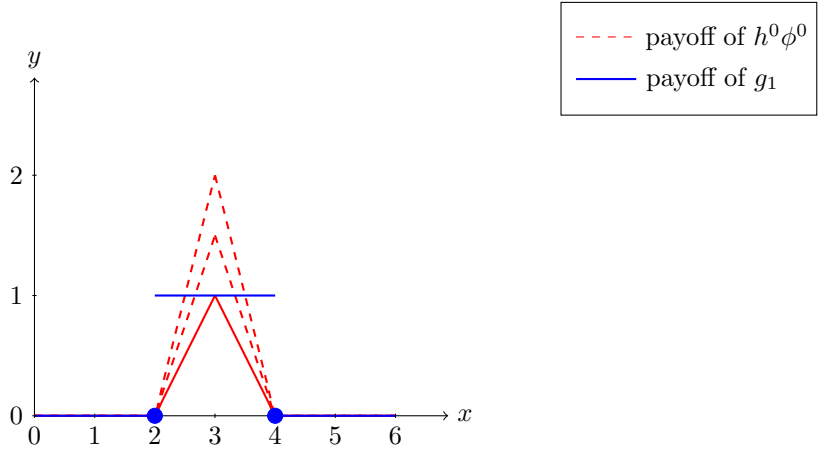
\begin{figure}[h]
\centering
\begin{tikzpicture}
  \draw[->] (0,0) -- (5.5,0) node[right] {$x$};
  \draw[->] (0,0) -- (0,4.5) node[above] {$y$};

   \draw[thick,scale=0.8,domain=0:3,smooth,variable=\x,red]  plot ({\x},{max(2*(\x-2),0)});
  \draw[thick,scale=0.8,domain=3:6,smooth,variable=\x,red]  plot ({\x},{max(2*(-\x+4),0});
  \draw[dashed, thick,scale=0.8,domain=0:3,smooth,variable=\x,red]  plot ({\x},{max(3*(\x-2),0)});
  \draw[dashed, thick,scale=0.8,domain=3:6,smooth,variable=\x,red]  plot ({\x},{max(3*(-\x+4),0});
  \draw[dashed, thick,scale=0.8,domain=0:3,smooth,variable=\x,red]  plot ({\x},{max(4*(\x-2),0)});
  \draw[dashed, thick,scale=0.8,domain=3:6,smooth,variable=\x,red]  plot ({\x},{max(4*(-\x+4),0});
  \draw[thick, scale=0.8 ,domain=9:10,smooth,variable=\x,blue] plot ({\x},5.6) node[right,black]{payoff of $g_1$};
  \draw[thick,scale=0.8,domain=0:2,smooth,variable=\x,blue]  plot ({\x},{0});
  \draw[thick,scale=0.8,domain=4:6,smooth,variable=\x,blue]  plot ({\x},{0});
  \draw[thick,scale=0.8,domain=2:4,smooth,variable=\x,blue]  plot ({\x},{2});

  \draw[dashed,scale=0.8 ,domain=9:10,smooth,variable=\x,red] plot ({\x},6.3) node[right,black]{payoff of $h^0\phi^0$};
  \draw (7,4) rectangle (10.4,5.5);

  \foreach \x in {0,1,2,3,4,5,6}
    \draw [scale=0.8](\x ,1pt) -- (\x ,-1pt) node[anchor=north] {$\x$};
  \foreach \y in {0,1,2}
    \draw [scale=0.8] (1pt,2*\y ) -- (-1pt,2*\y ) node[anchor=east] {$\y$};

    \draw [fill,blue] (1.6,0) circle [radius=0.1];
    \draw [fill,blue] (3.2,0) circle [radius=0.1];

\end{tikzpicture}
\caption{$h^{0}\protect\phi ^{0}$ does not dominate $g_{1}$ on $(K_{0},K_{0}+%
\protect\varepsilon )$ for any $\ h^{0}$ with $\protect\varepsilon =\protect%
\varepsilon (h^{0})$}
\end{figure}
\end{description}

We finally discuss the case $H^{1}<0$.\newline
This is, in general, a more expensive choice for the strategy $(H,h)$.
Indeed we have, for instance, that for $\widetilde{\omega }=K_{0}+1$, $%
H^{1}S^{1}(\widetilde{\omega })=H^{1}(K_{0}+1)<0$ while $g_{1}(\widetilde{%
\omega })=1$. Since for any strategy $(H,h)\in \mathbb{R}^{4}$, $V_{T}(H,h)(%
\widetilde{\omega })=H^{0}+H^{1}\widetilde{\omega }$ we need $H^{0}\geq
1-H^{1}(K_{0}+1)$, hence, the initial price $V_{0}(H,h)\geq
1-H^{1}(K_{0}+1-s_{0})$. By choosing the parameters $s_{0},K_{0}$ such that $%
K_{0}+1-s_{0}<0$ any superhedging strategy with $H^{1}<0$ is more expensive
than the trivial super-hedge given by $H^{0}=1,H^{1}=h^{0}=h^{0}=0$. Note
moreover that in order to cover the losses in $H^{1}S^{1}$ for large value
of $\omega $ we would need to take a long position in the option $\phi ^{1}$
(whose payoff dominates $S^{1}$) for an additional cost of $h^{1}c_{1}>0$
with $h^{1}>-H^{1}>0$. \newline

We can conclude that the cheapest super-replicating strategies are, in
general, given by $H^{0}S^{0}+H^{1}S^{1}$ with $H^{0},H^{1}\geq 0$ and it is
easy to see that
\begin{equation*}
\pi _{\Omega }(g_{1})=\min \left\{ \frac{s_{0}}{K_{0}},1\right\}=\frac{s_{0}%
}{K_{0}} >0.
\end{equation*}
\end{proof}

\section{Technical results and proofs}

\label{proofProp}

Recall that $\{\mathcal{F}_{t}\}_{t\in I}$ is the universal filtration which
satisfies in particular that $\mathcal{F}_{t}$ contains the family of
analytic sets of $(\Omega ,\mathcal{F}_{t}^{S})$ for any $t\in I$.

We indicate by $Mat(d\times (T+1);\mathbb{R})$ the space of $d\times (T+1)$
matrices with real entries representing the set of all the possible
trajectories of the price process: for every $\omega \in \Omega $ we have $%
(S_{0}(\omega ),S_{1}(\omega ),...,S_{T}(\omega ))\in Mat(d\times (T+1);%
\mathbb{R})$. Fix $t\leq T$: we indicate $S_{0:t}=(S_{0},S_{1},...,S_{t})$
and recall that $S_{0:t}^{-1}(A)=\{\omega \in \Omega \mid S_{0:t}(\omega
)\in A\}$ for $A\subset Mat(d\times (t+1);\mathbb{R}).$ We set $\Delta
S_{t}:=S_{t}-S_{t-1}$, $t=1,...,T.$

\subsection{ $\Omega _{\ast}$ and $\Omega _{\Phi}$ are analytic sets}

\begin{lemma}
\label{analytic}The set $\mathcal{P}_{f}=\{P\in \mathcal{P}\mid P\text{ has
finite support}\}$ is an analytic subset of $\mathcal{P}$ endowed with the
sigma-algebra generated by the $\sigma (\mathcal{P},C_{b})$ topology.
\end{lemma}

\begin{proof}
Set $E=\{\delta _{\omega }\mid \omega \in \Omega \}$ which is $\sigma (%
\mathcal{P},C_{b})$ closed (Th. 15.8 \cite{Aliprantis}) and observe that $%
\mathcal{P}_{f}$ is the convex hull of $E$. Consider for any $n\in \mathbb{N}
$ the simplex $\Delta _{n}\subset \mathbb{R}^{n}$ and the map
\begin{equation*}
\gamma _{n}:E^{n}\times \Delta _{n}\longrightarrow \mathcal{P}_{f}
\end{equation*}%
defined by $\gamma _{n}\left( \delta _{\omega _{1}},\ldots ,\delta _{\omega
_{n}},\lambda _{1},\dots ,\lambda _{n}\right) =\sum_{i=1}^{n}\lambda
_{i}\delta _{\omega _{i}}$ which is a continuous function in the product
topology. Since $E^{n}\times \Delta _{n}$ is closed in the product topology
of the Borel Space $\mathcal{P}^{n}\times \mathbb{R}^{n},$ then the image $%
\gamma _{n}\left( E^{n}\times \Delta _{n}\right) $ is analytic (Proposition
7.40 \cite{BS78}). Finally we notice that $\mathcal{P}_{f}=\bigcup_{n}\gamma
_{n}\left( E^{n}\times \Delta _{n}\right) $ which is therefore analytic,
being countable union of analytic sets.
\end{proof}

\begin{definition}
Let $\mathcal{L}^{\infty }(\Omega ,\mathcal{F}):=\left\{ f\in \mathcal{L}%
(\Omega ,\mathcal{F})\mid f\text{ is bounded}\right\} $. A subset $\mathcal{%
U\subset P}_{f}$ is countably determined if there exists a countable set $%
L\subseteq \mathcal{L}^{\infty }(\Omega ,\mathcal{F})$ such that
\begin{equation*}
\mathcal{U}:=\left\{ \mu \in \mathcal{P}_{f}\mid E_{\mu }[f]\leq 0,\forall
f\in L\right\}
\end{equation*}
\end{definition}

\begin{lemma}
If $\mathcal{U}\subseteq \mathcal{P}_{f}$ is countably determined then it is
analytic.
\end{lemma}

\begin{proof}
For each $f_{n}\in L$ define
\begin{equation*}
F_{n}:\mathcal{P}\rightarrow \mathbb{R}\text{ such that }F_{n}(\mu
)=\int_{\Omega }f_{n}d\mu .
\end{equation*}%
From Theorem 15.13 in \cite{Aliprantis}, $F_{n}$ is Borel measurable so that
\begin{equation*}
\mathcal{U}:=\left\{ \mu \in \mathcal{P}_{f}\mid E_{\mu }[f_{n}]\leq 0\text{
for all }n\in \mathbb{N}\right\} =\bigcap_{n\in \mathbb{N}%
}(F_{n})^{-1}(-\infty ,0]\cap \mathcal{P}_{f}
\end{equation*}%
is analytic, being countable intersection of analytic sets.
\end{proof}

\begin{lemma}
\label{PZ} Let $Z_{1}(\omega ):=\max_{i=1,\dots ,d}\max_{u=0,\dots
,T}|S_{u}^{i}(\omega )|$, $Z_{2}(\omega ):=\max_{j=1,\dots ,k}|\phi
^{j}(\omega )|$ and $Z=\max (Z_{1},Z_{2})$ then
\begin{eqnarray*}
\mathcal{P}_{Z} &=&\left\{ \mu \in \mathcal{P}_{f}\mid \exists \,Q\in
\mathcal{M}_{f}\text{ such that }\frac{dQ}{d\mu }=\frac{c(\mu )}{1+Z}\right\}
\\
\mathcal{P}_{Z,\Phi } &=&\left\{ \mu \in \mathcal{P}_{f}\mid \exists \,Q\in
\mathcal{M}_{\Phi }\text{ such that }\frac{dQ}{d\mu }=\frac{c(\mu )}{1+Z}%
\right\}
\end{eqnarray*}%
are analytic subsets of $\mathcal{P}$ where $c(\mu )=E_{\mu }\left[
(1+Z)^{-1}\right] ^{-1}$.
\end{lemma}

\begin{proof}
Assume $\mathcal{P}_{Z}\neq \varnothing $ (resp. $\mathcal{P}_{Z,\Phi }\neq
\varnothing $) otherwise there is nothing to prove. Fix any $t\in \left\{
1,...,T\right\} $. Let $Mat(d\times t;\mathbb{Q})$ be the countable set of $%
d\times t$ matrices with rational entries and denote its elements by $q_{n},$
$n\in \mathbb{N}$. For $q_{n}\in Mat(d\times t;\mathbb{Q})$, consider the
set $\{A_{n,m}\}$ with $A_{n,m}=\{\omega \in \Omega \mid S_{0:t-1}\in
B_{1/m}( q_{n})\}\in \mathcal{F}_{t-1}$, where $B_{1/m}( q_{n})$ denotes the
ball (in the Euclidean norm of $Mat(d\times t;\mathbb{R})$) with radius $1/m$
centered in $q_n$. Define
\begin{eqnarray}
f_{n,m}^{i}:= &&\left( \frac{S_{t}^{i}-S_{t-1}^{i}}{1+Z}\right) \mathbf{1}%
_{A_{n,m}}\in \mathcal{L}^{\infty }(\Omega ,\mathcal{F}),  \notag \\
g^{j}:= &&\left( \frac{\phi ^{j}}{1+Z}\right) \in \mathcal{L}^{\infty
}(\Omega ,\mathcal{F}).  \label{gj}
\end{eqnarray}%
The following sets
\begin{eqnarray*}
\mathcal{U}:= &&\left\{ \mu \in \mathcal{P}_{f}\mid E_{\mu
}[f_{n,m}^{i}]=0\;\forall i,n,m\right\} \\
\mathcal{U}_{\Phi }:= &&\left\{ \mu \in \mathcal{P}_{f}\mid E_{\mu
}[f_{n,m}^{i}]=0\text{ and }E_{\mu }[g^{j}]=0\;\forall i,n,m,j\right\}
\end{eqnarray*}%
are analytic since they are countably determined. We now show that $\mathcal{%
U}=\mathcal{P}_{Z}$ and $\mathcal{U}_{\Phi }=\mathcal{P}_{Z,\Phi }$ and this
will complete the proof. \newline
For any fixed $\mu \in \mathcal{U}$ we have by construction:
\begin{equation}
\int_{\Omega }\frac{S_{t}^{i}}{1+Z}\mathbf{1}_{A_{n,m}}d\mu =\int_{\Omega }%
\frac{S_{t-1}^{i}}{1+Z}\mathbf{1}_{A_{n,m}}d\mu \quad \text{ for every }%
A_{n,m}.  \label{integral}
\end{equation}%
Consider the finite set of matrices $\{s_{j}\}_{j=1}^{h}:=\{S_{0:t-1}(\omega
)\in Mat(d\times t;\mathbb{R})\mid \omega \in supp(\mu )\}$ where $h=h(\mu )$
depends on $\mu $. For every $j=1,\ldots,h$ there exists $q_{n(j)},m(j)$
such that $s_j\in B_{1/m(j)}(q_{n(j)})$ and the balls $B_{1/m(j)}(q_{n(j)})$
are all disjoint. Therefore $A_{n(j),m(j)}$ is such that
\begin{equation*}
\mu (B_{j})=\mu \left( A_{n(j),m(j)}\right)
\end{equation*}%
where $B_{j}:=\{S_{0:t-1}=s_{j}\}$. Since $\{B_{j}\}_{j=1}^h$ are atoms for $%
\mu$ in $\mathcal{F}_{t-1}$, we conclude that
\begin{equation*}
\int_{\Omega }\frac{S_{t}^{i}}{1+Z}\mathbf{1}_{B_{j}}d\mu =\int_{\Omega }%
\frac{S_{t-1}^{i}}{1+Z}\mathbf{1}_{B_{j}}d\mu \text{\quad for every }%
j=1,\dots ,h
\end{equation*}%
and $E_{\mu }\left( \frac{S_{t}^{i}}{1+Z}\mid \mathcal{F}_{t-1}\right)
=E_{\mu }\left( \frac{S_{t-1}^{i}}{1+Z}\mid \mathcal{F}_{t-1}\right) $.
Define $Q$ by $\frac{dQ}{d\mu }:=\frac{c}{1+Z}$ where $c:=c(\mu )>0$ is the
normalization constant. Then $,$ $Q\sim \mu $, $Q\in \mathcal{P}_{f}$ and:
\begin{equation}
E_{\mu }\left( \frac{S_{t}^{i}}{1+Z}\mid \mathcal{F}_{t-1}\right) =E_{\mu
}\left( \frac{S_{t-1}^{i}}{1+Z}\mid \mathcal{F}_{t-1}\right) \text{ if and
only if }E_{Q}\left( S_{t}^{i}\mid \mathcal{F}_{t-1}\right) =S_{t-1}^{i}.
\label{changemeasure}
\end{equation}%
Thus we can conclude $Q\in \mathcal{M}_{f}$ and $\mathcal{U}\subseteq
\mathcal{P}_{Z}$. Take now $\mu \in \mathcal{P}_{Z}$ then there exists $Q$
such that $E_{Q}\left( S_{t}^{i}\mid \mathcal{F}_{t-1}\right) =S_{t-1}^{i}$
and $\frac{dQ}{d\mu }=\frac{c}{1+Z}$. From Equation \eqref{changemeasure} we
have that condition \eqref{integral} holds and hence $\mu \in \mathcal{U}$.
\newline
Recall that $\mathcal{M}_{\Phi }$ is defined in (\ref{Mphi}) and consider
now $\mu \in \mathcal{U}_{\Phi }\subseteq \mathcal{U}$. Then there exists $%
Q\in \mathcal{M}_{f}$ such that $\frac{dQ}{d\mu }=\frac{c(\mu )}{1+Z}$.
Moreover $E_{\mu }[g^{j}]=0$ for every $j=1,\dots ,k$ so that, by (\ref{gj}%
), $E_{Q}[\phi ^{j}]=0$. In this way $\mathcal{U}_{\Phi }\subseteq \mathcal{P%
}_{Z,\Phi }$. Take now $\mu \in \mathcal{P}_{Z,\Phi }$ then $\mu \in
\mathcal{P}_{Z}$ from the previous part of the proof. Moreover there exists $%
Q\in \mathcal{M}_{\Phi }$ such that $E_{Q}\left[ \phi ^{j}\right] =0$ and $%
\frac{dQ}{d\mu }=\frac{c}{1+Z}$. Again by (\ref{gj}) we have $E_{\mu
}[g^{j}]=0$ for every $j=1,\dots ,k$ and hence $\mu \in \mathcal{U}_{\Phi }$.
\end{proof}

\begin{proposition}
\label{propAn} $\Omega _{\ast }$ and $\Omega _{\Phi }$ are
analytic subsets of $(\Omega ,\mathcal{F})$.
\end{proposition}

\begin{proof}
Consider the Baire space $\mathbb{N}^{\mathbb{N}}$ of all
sequences of natural numbers. In this proof we denote by
$B_{\varepsilon }(\omega )$ the
closed ball of radius $\varepsilon $, centered in $\omega $ in $(\Omega ,d)$%
. \newline Consider a dense subset $\{\omega _{i}\}_{i=1}^{\infty
}$ of $\Omega $. For any $\mathbf{n}=(n_{1},...,n_{k},...)\in
\mathbb{N}^{\mathbb{N}}$ we denote
by $\mathbf{n}(1),\dots ,\mathbf{n}(k)$ the first $k$ terms (i.e. $%
n_{1},...,n_{k}$). Define
\begin{equation*}
A_{\mathbf{n}(1)}:=B_{1}(\omega _{\mathbf{n}(1)}).
\end{equation*}%
Let now $\{\omega _{\mathbf{n}(1),i}\}_{i=1}^{\infty }$ a dense subset of $%
A_{\mathbf{n}(1)}$ we define
\begin{equation*}
A_{\mathbf{n}(1),\mathbf{n}(2)}:=B_{\frac{1}{2}}(\omega _{\mathbf{n}(1),%
\mathbf{n}(2)})\cap A_{\mathbf{n}(1)}.
\end{equation*}%
At the $k^{th}$ step we shall have $\{\omega _{\mathbf{n}(1),\ldots ,\mathbf{%
n}(k-1),i}\}_{i=1}^{\infty }$ a dense subset of $A_{\mathbf{n}(1),\ldots ,%
\mathbf{n}(k-1)}$ and we define the closed set
\begin{equation*}
A_{\mathbf{n}(1),\ldots ,\mathbf{n}(k)}:=B_{\frac{1}{k}}(\omega _{\mathbf{n}%
(1),\ldots ,\mathbf{n}(k)})\cap A_{\mathbf{n}(1),\ldots
,\mathbf{n}(k-1)}.
\end{equation*}%
Notice that for any $\omega \in \Omega $ there will exists an
$\mathbf{n}\in \mathbb{N}^{\mathbb{N}}$ such that
\begin{equation}
\bigcap _{k\in \mathbb{N}}A_{\mathbf{n}(1),\ldots
,\mathbf{n}(k)}=\{\omega \}.  \label{intA}
\end{equation}%
We consider the \emph{nucleus} of the Souslin scheme given by

\begin{equation}
\bigcup_{\mathbf{n}\in \mathbb{N}^{\mathbb{N}}}\bigcap_{k\in \mathbb{N}}A_{%
\mathbf{n}(1),\ldots ,\mathbf{n}(k)}\times \{Q\in \mathcal{P}_{Z}\mid Q(A_{%
\mathbf{n}(1),\ldots ,\mathbf{n}(k)})>0\}  \label{nucleus}
\end{equation}%
Observe that $A_{\mathbf{n}(1),\ldots ,\mathbf{n}(k)}$ closed in
$\Omega $
implies $\{Q\in \mathcal{P}\mid Q(A_{\mathbf{n}(1),\ldots ,\mathbf{n}%
(k)})\geq \frac{1}{m}\}$ is $\sigma (\mathcal{P},C_{b})$-closed
from Corollary 15.6 in \cite{Aliprantis}. Therefore
\begin{equation*}
\{Q\in \mathcal{P}\mid Q(A_{\mathbf{n}(1),\ldots ,\mathbf{n}%
(k)})>0\}=\bigcup_{m}\{Q\in \mathcal{P}\mid Q(A_{\mathbf{n}(1),\ldots ,%
\mathbf{n}(k)})\geq \frac{1}{m}\}
\end{equation*}%
is Borel measurable in $(\mathcal{P},\sigma (\mathcal{P},C_{b}))$. By Lemma %
\ref{PZ} we have that $\{Q\in \mathcal{P}_{Z}\mid Q(A_{\mathbf{n}(1),\ldots ,%
\mathbf{n}(k)})>0\}$ is analytic. We can conclude that $A_{\mathbf{n}%
(1),\ldots ,\mathbf{n}(k)}\times \{Q\in \mathcal{P}_{Z}\mid Q(A_{\mathbf{n}%
(1),\ldots ,\mathbf{n}(k)})>0\}$ is an analytic subset of $\Omega
\times \mathcal{P}$ (which is a Polish space). \newline From Lemma
\ref{PZ} we observe that any $\mu \in \mathcal{P}_{Z}$ admits an
equivalent martingale measure with finite support. From $\Omega
_{\ast
}=\left\{ \omega \in \Omega \mid \exists Q\in \mathcal{M}_{f}\text{ s.t. }%
Q(\omega )>0\right\} $, if $\omega \notin \Omega _{\ast }$ then
$\omega \notin supp(\mu )$ for any $\mu \in \mathcal{P}_{Z}$.
Taking (\ref{intA})
into account, if $\omega \notin \Omega _{\ast }$ we can find a large enough $%
\bar{k}$ such that $A_{\mathbf{n}(1),\ldots
,\mathbf{n}(\bar{k})}\cap supp(\mu )=\varnothing $. We then have
\begin{equation}
\bigcap_{k\in \mathbb{N}}A_{\mathbf{n}(1),\ldots
,\mathbf{n}(k)}\times
\{Q\in \mathcal{P}_{Z}\mid Q(A_{\mathbf{n}(1),\ldots ,\mathbf{n}%
(k)})>0\}=\left\{
\begin{array}{cc}
\{\omega \}\times \mathcal{P}_{\omega } & \text{ if }\omega \in
\Omega
_{\ast } \\
\varnothing  & \text{ if }\omega \notin \Omega _{\ast }%
\end{array}%
\right. ,  \label{int:k}
\end{equation}%
where $\mathcal{P}_{\omega }=\{Q\in \mathcal{P}_{Z}\mid Q(\{\omega
\})>0\}$.
\newline
From Proposition 7.35 and Proposition 7.41 in \cite{BS78} any
kernel of a Souslin scheme of analytic sets is again an analytic
set. Then
\begin{equation*}
\bigcup_{\mathbf{n}\in \mathbb{N}^{\mathbb{N}}}\bigcap_{k\in \mathbb{N}}A_{%
\mathbf{n}(1),\ldots ,\mathbf{n}(k)}\times \{Q\in \mathcal{P}_{Z}\mid Q(A_{%
\mathbf{n}(1),\ldots ,\mathbf{n}(k)})>0\}
\end{equation*}%
is an analytic set in $\Omega \times \mathcal{P}$ whose projection on $%
\Omega $, thanks to \eqref{int:k}, is equal to $\Omega _{\ast }$.
Since the projection $\Pi :\Omega \times \mathcal{P}\rightarrow
\Omega $ is continuous we finally deduce that $\Omega _{\ast }$ is
analytic. \newline For $\Omega _{\Phi }$ repeat the same proof
replacing $\mathcal{P}_{Z}$ with $\mathcal{P}_{Z,\Phi }$.
\end{proof}

\begin{remark}
Let $\hat{\Omega}\subseteq \Omega $ be an analytic subset of $(\Omega ,%
\mathcal{F})$. An inspection of the proof shows that
\begin{align}
\hat{\Omega}_{\ast }:=& \left\{ \omega \in \hat{\Omega}\mid
\exists Q\in \mathcal{M}_{f}\text{ s.t. }Q(\hat{\Omega})=1\text{
and }Q(\omega
)>0\right\}   \label{omegaHat*} \\
\hat{\Omega}_{\Phi }:=& \left\{ \omega \in \hat{\Omega}\mid
\exists Q\in \mathcal{M}_{\Phi }\text{ s.t.
}Q(\hat{\Omega})=1\text{ and }Q(\omega )>0\right\}   \notag
\end{align}%
are also analytic subsets of $(\Omega ,\mathcal{F})$. Indeed, $\mathcal{P}_{%
\hat{\Omega}}:=\{P\in \mathcal{P}\mid P(\hat{\Omega})=1\}$ is an
analytic
subset of $\mathcal{P}$, by Proposition 7.43 in \cite{BS78}, therefore $%
\mathcal{P}_{Z}\cap \mathcal{P}_{\hat{\Omega}}$ is analytic and
one may replace in the above proof $\mathcal{P}_{Z}$ with
$\mathcal{P}_{Z}\cap \mathcal{P}_{\hat{\Omega}}$ and $\Omega
_{\ast }$ with $\hat{\Omega}_{\ast }$ to obtain the conclusion.
\end{remark}

\begin{remark}
In one-period markets ($T=1$), $\Omega _{\ast }$ is a Borel
measurable set. To see this observe that if there are no one point
arbitrages then $\Omega
_{\ast }=\Omega \in \mathcal{B}(\Omega )$ by Corollary 4.11 in \cite{BFM14}%
. When this condition is violated, there exists a strategy $H^{1}\in \mathbb{%
R}^{d}$ such that $H^{1}\cdot (S_{1}-S_{0})\geq 0$ and
$B^{1}:=\{\omega \in \Omega \mid H^{1}\cdot (S_{1}(\omega
)-S_{0})>0\}$ is non-empty and Borel
measurable. Indeed $B^{1}=(f\circ S_{1})^{-1}(0,\infty )$ with $%
f(x):=H^{1}\cdot (x-S_{0})$ continuous and $S_{1}$ Borel
measurable. Observe now that, restricted to the set $\Omega
\setminus B^{1}$, one asset is
redundant (say $S^{d}$) so that the market can be described by $%
(S^{0},\ldots ,S^{d-1})$. If there is no one point arbitrage we
have $\Omega _{\ast }=\Omega \setminus B^{1}\in \mathcal{B}(\Omega
)$. Otherwise we can iteratively repeat the same argument to
construct $B^{i}:=\{\omega \in \Omega \setminus \cup
_{j=1}^{i-1}B^{j}\mid H^{i}\cdot (S_{1}(\omega )-S_{0})>0\}\in
\mathcal{B}(\Omega )$ and dropping iteratively one
additional asset. Since the number of assets is finite the procedure takes $%
\beta \leq d$ steps. On the resulting set there are no one point
arbitrages
so that $\Omega _{\ast }=(\cup _{i=1}^{\beta }B^{j})^{C}\in \mathcal{B}%
(\Omega )$.
\end{remark}

\subsection{On the key Proposition \protect\ref{supsup}}

\begin{remark}
\label{remarkutile} We point out at this stage that $\Omega _{\ast }$ is not
only analytic but also it belongs to $\mathcal{F}_{T}$ where $\mathcal{F}%
_{T} $ is the universal completion of $\sigma (S_{t}\mid t\leq T)$. Indeed $%
\Omega _{\ast }\subseteq S_{0:T}^{-1}(S_{0:T}(\Omega _{\ast }))$. Moreover
for any $\omega _{1}\in S_{0:T}^{-1}(S_{0:T}(\Omega _{\ast }))$ there exists
$\omega _{2}\in \Omega _{\ast }$ such that $S_{0:T}(\omega
_{1})=S_{0:T}(\omega _{2})$. Therefore for any $Q\in \mathcal{M}_{f}$ such
that $Q(\{\omega _{2}\})>0$ and $Q(\{\omega _{1}\})=0$, the measure $\tilde{Q%
}$ such that $\tilde{Q}(\{\omega _{1}\}):=Q(\{\omega _{2}\})$, $\tilde{Q}%
(\{\omega _{2}\}):=0$ and $\tilde{Q}=Q$ elsewhere is a martingale measure.
Necessarily $\omega _{1}\in \Omega _{\ast }$.
\end{remark}

In the proof of Proposition \ref{supsup} we will make use of the following
simple fact: set $\Omega _{\ast }^{T}:=\Omega _{\ast }\in \mathcal{F}_{T}$
then by backward recursion we have
\begin{equation*}
\Omega _{\ast }^{t}:=S_{0:t}^{-1}(S_{0:t}(\Omega _{\ast }^{t+1}))\in
\mathcal{F}_{t},\quad \Omega _{\ast }^{t+1}\subseteq \Omega _{\ast }^{t}%
\text{ for any }t=0,\ldots ,T-1,\text{ and}\quad \Omega _{\ast
}=\bigcap_{t=1}^{T}\Omega _{\ast }^{t}\text{.}
\end{equation*}%
Notice that $\Omega _{\ast }^{t}$ can be interpreted as the $\mathcal{F}_{t}$%
-measurable projection of $\Omega _{\ast }$ since $\Omega _{\ast
}^{t}=S_{0:t}^{-1}(S_{0:t}(\Omega _{\ast }))$. \newline


We also recall that the condition No one point arbitrage holds true on $%
\Omega_*$. If indeed there exists $H\in\mathcal{H}$ such that $(H\cdot
S)_T\geq 0$ with $(H\cdot S)_T(\omega)>0$ for some $\omega\in\Omega_*$, then
any measure $P$ such that $P(\omega)>0$ cannot be a martingale measure,
which contradicts \eqref{omega*}.

\subsubsection{Proof of Proposition \protect\ref{supsup}}

We show, in several steps, that $\pi _{\ast }(g)=\sup_{Q\in \mathcal{M}%
_{f}}\pi _{Q}(g)$ where $\pi _{\ast }$ and $\pi _{Q}$ are defined in (\ref%
{pi}) and (\ref{piQ}) and $g\in \mathcal{L}(\Omega ,\mathcal{F})$.

\textbf{Step 1:} The first step is to construct, for any $1\leq t\leq T$, an
$\mathcal{F}_{t-1}$-measurable random set $R_{t,X,D}\subseteq \mathbb{R}%
^{d+1} $ whose interpretation is the following: if $\omega$ occurs, any $%
H^1,\ldots H^d,H^{d+1}\in R_{t,X,D}(\omega)$ represents a strategy at time $%
t-1$ that allows to super-hedge the random variable $X$ at time $t$, for any
trajectory in $D\subseteq \Omega$. Here $H^{d+1}$ represents the investment
in the non-risky asset. Note that we need to incorporate the additional
feature given by the choice of the set $D$ since we want to super-hedge the
random variable $g$ only on $\Omega_*\subseteq \Omega$. \bigskip

Recall $\Delta S_{t}=S_{t}-S_{t-1}$. Consider, for an arbitrary $1\leq t\leq
T$, $D\in \mathcal{F}_{t}$ and $X\in \mathcal{L}(\Omega ,\mathcal{F}),$ the
multifunction
\begin{equation*}
\psi _{t,X,D}:\omega \mapsto \left\{ \left[ \Delta S_{t}(\widetilde{\omega }%
);1;X(\widetilde{\omega })\right] \mathbf{1}_{D}\mid \widetilde{\omega }\in
\Sigma _{t-1}^{\omega }\right\} \subseteq \mathbb{R}^{d+2}
\end{equation*}%
where $\left[ \Delta S_{t};1;X\right] \mathbf{1}_{D}=\left[ \Delta S_{t}^{1}%
\mathbf{1}_{D},\ldots ,\Delta S_{t}^{d}\mathbf{1}_{D},\mathbf{1}_{D},X%
\mathbf{1}_{D}\right] $ and $\Sigma _{t-1}^{\omega }$ is the level set of
the trajectory $\omega $ up to time $t-1$ i.e. $\Sigma _{t-1}^{\omega }=\{%
\widetilde{\omega }\in \Omega \mid S_{0:t-1}(\widetilde{\omega }%
)=S_{0:t-1}(\omega )\}$. We show that $\psi _{t,X,D}$ is an $\mathcal{F}%
_{t-1}$-measurable multifunction. Indeed we need to show that, for any open
set $O\subseteq \mathbb{R}^{d}\times \mathbb{R}^{2}$,
\begin{equation*}
\{\omega \in \Omega \mid \psi _{t,X,D}(\omega )\cap O\neq \varnothing
\}=S_{0:t-1}^{-1}\left( S_{0:t-1}\left( B\right) \right) \in \mathcal{F}%
_{t-1}\quad \text{where }B=(\left[ \Delta S_{t};1;X\right] \mathbf{1}%
_{D})^{-1}(O).
\end{equation*}%
First $\left[ \Delta S_{t},1,X\right] \mathbf{1}_{D}$ is an $\mathcal{F}$%
-measurable random vector then $B\in \mathcal{F}$. Second $S_{u}$ is a Borel
measurable function for any $0\leq u\leq t-1$ so that we have, as a
consequence of Theorem III.18 in \cite{DM82}, that $S_{0:t-1}(B)$ belongs to
the sigma-algebra generated by the analytic sets in $Mat(d\times t;\mathbb{R}%
)$ endowed with its Borel sigma-algebra. Applying now Theorem III.11 in \cite%
{DM82} we deduce that $S_{0:t-1}^{-1}(S_{0:t-1}(B))\in \mathcal{F}_{t-1}$
and hence the desired measurability for $\psi _{t,X,D}$.\newline
By preservation of measurability (see \cite{R} for instance) the
multifunction
\begin{equation*}
\psi _{t,X,D}^{\ast }(\omega ):=\left\{ H\in \mathbb{R}^{d+2}\mid H\cdot
y\leq 0\quad \forall y\in \psi _{t,X,D}(\omega )\right\}
\end{equation*}%
is also $\mathcal{F}_{t-1}$-measurable and thus, the same holds true for $%
-\psi _{t,X,D}^{\ast }\cap \{\mathbb{R}^{d+1}\times \{-1\}\}$. The
projection on the first $d+1$ components, $R_{t,X,D}:=\Pi _{x_{1},\ldots
,x_{d+1}}(-\psi _{t,X,D}^{\ast }\cap \{\mathbb{R}^{d+1}\times \{-1\}\})$,
provides the building blocks for the super-replicating strategy for $X$. By
the previous construction we have indeed that
\begin{equation}
R_{t,X,D}(\omega )=\left\{ H\in \mathbb{R}^{d+1}\mid H^{d+1}\mathbf{1}%
_{D}+\sum_{i=1}^{d}H^{i}\Delta S_{t}^{i}(\widetilde{\omega })\mathbf{1}%
_{D}\geq X(\widetilde{\omega })\mathbf{1}_{D}\quad \forall \widetilde{\omega
}\in \Sigma _{t-1}^{\omega }\right\}   \label{Replicating}
\end{equation}%
Notice that if $D\cap \Sigma _{t-1}^{\omega }=\varnothing $ then $%
R_{t,X,D}(\omega )=\mathbb{R}^{d+1}$. Note also that $R_{t,X,D}$ is, by
construction, a closed set.

Denote by $\Pi _{x_{d+1}}(R_{t,X,D})$ the projection on the $(d+1)$-th
component, which is a random interval in $\mathbb{R}$ with possible values $%
\{\varnothing \},\{\mathbb{R}\}$. Observe now that the projection is
continuous and that the infimum of a real-valued random set $A$ preserve the
measurability since
\begin{equation*}
\left\{ \omega \in \Omega \mid \inf \{a\mid a\in A(\omega )\}<y\right\}
=\left\{ \omega \in \Omega \mid A(\omega )\cap (-\infty ,y)\neq \varnothing
\right\}
\end{equation*}%
Conclude, therefore, that $X_{t-1}:=\inf \Pi _{x_{d+1}}(R_{t,X,D})$ is an $%
\mathcal{F}_{t-1}$-measurable function with values in $\mathbb{R}\cup \{\pm
\infty \}$.

\textbf{Step 2}. We prove that for every $\omega \in \{|X_{t-1}|<\infty \}$
the infimum in $X_{t-1}$ is actually a minimum. To this aim fix $\omega \in
\{|X_{t-1}|<\infty \}$ and notice that there might exist $L\in \mathbb{R}%
^{d}\setminus \{0\}$ such that $L\cdot \Delta S_{t}=0$ on $\Sigma
_{t-1}^{\omega }\cap \Omega _{\ast }^{t}$, meaning that some assets are
redundant on this level set. We can reduce the number of assets by selecting
$i_{1},\ldots ,i_{k}\in (1,...,d)$ such that $l_{1}\Delta
S_{t}^{i_{1}}+\ldots +l_{k}\Delta S_{t}^{i_{k}}=0$ implies $l_{j}=0$ for
every $j=1,\ldots ,k$. Consider the closed set
\begin{equation*}
\widetilde{R}(\omega )=\left\{ H\in R_{t,X,D}(\omega )\mid H^{i_{j}}=0\text{
for every }j=1,\dots ,k\right\}
\end{equation*}%
and observe that
\begin{eqnarray*}
X_{t-1}(\omega ) &=&\inf \Pi _{x_{d+1}}\left( R_{t,X,D}(\omega )\right)
=\inf \Pi _{x_{d+1}}(\widetilde{R}(\omega )) \\
&=&\inf \Pi _{x_{d+1}}\left( \widetilde{R}(\omega )\cap \left\{ \mathbb{R}%
^{d}\times \lbrack X_{t-1}(\omega ),X_{t-1}(\omega )+1]\right\} \right) .
\end{eqnarray*}

The set $Ko(\omega ):=\widetilde{R}(\omega )\cap \left\{ \mathbb{R}%
^{d}\times \lbrack X_{t-1}(\omega ),X_{t-1}(\omega )+1]\right\} $ is closed
being the intersection of closed sets. We claim that $Ko(\omega )$ is
bounded. By contradiction, suppose it is unbounded. Let $\hat{H}%
_{n}=(H_{n},H_{n}^{d+1})\in Ko(\omega )\subset \mathbb{R}^{d}\times \mathbb{R%
}$, such that $\Vert H_{n}\Vert \rightarrow +\infty $. By definition $%
H_{n}^{i_{j}}=0$ for every $j=1,\dots ,k$ and $H_{n}^{d+1}$ is bounded by $%
X_{t-1}(\omega )+1$. For any $\widetilde{\omega }\in D\cap \Sigma
_{t-1}^{\omega }$ and any $n$ we have
\begin{equation*}
\dfrac{X_{t-1}(\omega )+1}{\Vert H_{n}\Vert }+\dfrac{H_{n}}{\Vert H_{n}\Vert
}\cdot \Delta S_{t}(\widetilde{\omega })\geq \dfrac{X_{t}(\widetilde{\omega }%
)}{\Vert H_{n}\Vert }.
\end{equation*}%
Since $\frac{H_{n}}{\Vert H_{n}\Vert }$ lies on the unit sphere of $\mathbb{R%
}^{d}$, we can extract a subsequence converging to $H^{\ast }$ with $\Vert
H^{\ast }\Vert =1$. Therefore passing to the limit over this subsequence we
have $H^{\ast }\cdot \Delta S_{t}(\widetilde{\omega })\geq 0$ for every $%
\widetilde{\omega }\in D\cap \Sigma _{t-1}^{\omega }$. From No one point
arbitrage condition we deduce $H^{\ast }\cdot \Delta S_{t}=0$ on $D\cap
\Sigma _{t-1}^{\omega }$. Since $H_{n}\in Ko(\omega )$ then $(H^{\ast
})^{i_{j}}=0$ on the redundant assets and thus $H^{\ast }=0$ which is a
contradiction. \newline
The set $Ko(\omega )$ is closed and bounded in $\mathbb{R}^{d+1}$, hence
compact. From the continuity of the projection $\Pi _{x_{d+1}}(Ko(\omega ))$
is compact, so that the infimum is attained.

\bigskip

\textbf{Step 3:} We now provide a backward procedure which yields the
super-replication price and the corresponding optimal strategy. By classical
arguments, when we fix a reference probability $Q\in \mathcal{M}_{f}$ this
procedure yields two processes $X_{t}(Q)$ and $H_{t}(Q)$ such that
\begin{equation}
g\leq \sum_{u=t+1}^{T}H_{u}(Q)\cdot \Delta
S_{u}+X_{t}(Q)=\sum_{t=1}^{T}H_{t}(Q)\cdot \Delta S_{t}+X_{0}(Q)\quad Q-%
\text{a.s.}  \label{splitQ}
\end{equation}%
where $X_{t}(Q)$ represents the minimum amount of cash that we need at time $%
t$ in order to super-hedge $g$ in the $Q$-a.s. sense. Recall that from $%
NA(Q) $ we necessarily have $X_{t}(Q)>-\infty $ on supp$(Q)$. With no loss
of generality set $X_{t}(Q)(\omega )=-\infty $ for any $\omega \notin \text{%
supp}(Q)$. Now we prove the pathwise counterpart of \eqref{splitQ}:

\bigskip

Set $X_{T}:=g$ and $D_{T}:=\Omega _{\ast }$ which belongs to $\mathcal{F}%
_{T} $ by Remark \ref{remarkutile} and consider first the random set $%
R_{T,X_T,D_T}$. The random variable $X_{T-1}:=\inf \Pi
_{x_{d+1}}(R_{T,X_{T},D_{T}})$ represents the minimum amount of cash that we
need at time $T-1$ in order to super-hedge $g$ on $\Omega_*$. $X_{T-1}$ is
therefore the $\mathcal{F}_{T-1}$-measurable random variable that needs to
be super-replicated at time $T-2$.\newline
For $t=T-1,\ldots ,0$ we indeed iterate the procedure by taking $X_{t}:=\inf
\Pi _{x_{d+1}}(R_{t+1,X_{t+1},D_{t+1}})$, $%
D_{t}=S_{0:t}^{-1}(S_{0:t}(D_{t+1}))\in \mathcal{F}_{t}$ and the random set $%
R_{t+1,X_{t+1},D_{t+1}}$ as defined before. We again have that $X_{t}$ is an
$\mathcal{F}_{t}$-measurable function with values in $\mathbb{R}\cup \{\pm
\infty \}$. \newline

This backward procedure yields the super-hedging price $X_{0}$ on $\Omega
_{\ast }$ but also provide the corresponding cheapest portfolio as follows:
note first that for every $\omega \in \Omega _{\ast }$, $X_{t}(\omega
)>-\infty $. If this is not the case there exists a sequence $%
(H_{n},x_{n})_{n\in \mathbb{N}}\in \mathbb{R}^{d}\times \mathbb{R}$ such
that $x_{n}\downarrow -\infty $, $x_{n}+H_{n}\Delta S_{t+1}(\widetilde{%
\omega })\geq X_{t+1}(\widetilde{\omega })$ for every $\widetilde{\omega }%
\in D_{t+1}\cap \Sigma _{t}^{\omega }$ and hence $Q$-a.s. for every $Q\in
\mathcal{M}_{f}$ such that $Q(\Sigma _{t}^{\omega })>0$. This would lead to
a contradiction with $X_{t}(Q)>-\infty $. From now on we therefore assume
that $X_{t}(\omega )>-\infty $. In the case $X_{t}(\omega )<\infty $ for
every $t=0,\ldots ,T-1$, Step 2 provides that $X_{t}$ is actually a minimum.
The $\mathcal{F}_{t}$-measurable multifunction given by $\Pi _{x_{1},\ldots
,x_{d}}(R_{t+1,X_{t+1},D_{t+1}}\cap \left\{ \mathbb{R}^{d}\times
X_{t}\right\} )$ is therefore non-empty for every $t=0,\ldots ,T-1$ and thus
admits a measurable selector $H_{t+1}$. The strategy $H_{1},\ldots ,H_{T}$
satisfy the inequalities
\begin{eqnarray*}
g\leq H_{T}\cdot \Delta S_{T}+X_{T-1} &&\text{ on }D_{T} \\
X_{T-1}\leq H_{T-1}\cdot \Delta S_{T-1}+X_{T-2} &&\text{ on }D_{T-1} \\
&\ldots & \\
X_{1}\leq H_{1}\cdot \Delta S_{1}+X_{0} &&\text{ on }D_{1}
\end{eqnarray*}%
and it represents a super-hedge on $\Omega _{\ast }=\bigcap_{t=1}^{T}D_{t}$
as
\begin{equation}
g\leq H_{T}\cdot \Delta S_{T}+X_{T-1}\leq \sum_{t=T-1}^{T}H_{t}\cdot \Delta
S_{t}+X_{T-2}\leq \ldots \leq \sum_{t=1}^{T}H_{t}\cdot \Delta S_{t}+X_{0}
\label{split}
\end{equation}%
holds true for any $\omega \in \Omega _{\ast }$. When instead $X_{t}(\omega
)=\infty $ for some $\omega \in \Omega _{\ast }$ and for some $t\geq 0$ then
by simply taking $X_{u}\equiv \infty $ and $H_{u}$ arbitrary for every $%
u\leq t$, the inequality \eqref{split} is trivially satisfied.

\bigskip

\textbf{Step 4:} In order to prove (\ref{piSup}) we recursively show that $%
X_{t}(\omega )=\sup_{Q\in \mathcal{M}_{f}}X_{t}(Q)(\omega )$ for any $\omega
\in \Omega _{\ast }$ which, in particular, implies $X_{0}=\sup_{Q\in
\mathcal{M}_{f}}X_{0}(Q)$. Obviously $X_{t}(\omega )\geq X_{t}(Q)(\omega )$
for any $\omega \in \Omega _{\ast }$ so that $X_{t}\geq \sup_{Q\in \mathcal{M%
}_{f}}X_{t}(Q)$. Thus, we need only to prove the reverse inequality.

\bigskip

For $t=T$ the claim is obvious: $X_{T}=g$. By backward recursion suppose now
it holds true for any $u$ with $t+1\leq u\leq T$ i.e. $X_{u}(\omega
)=\sup_{Q\in \mathcal{M}_{f}}X_{u}(Q)(\omega )$ for any $\omega \in \Omega
_{\ast }$.\newline
From the recursive hypothesis in order to find a super-replication strategy
with the same price for any $Q\in \mathcal{M}_{f}$ we need to
super-replicate $X_{t+1}$. We fix a level set $\Sigma _{t}^{\omega }$ and
recall that $X_{t}$ is $\mathcal{F}_{t}$-measurable, hence it is constant on
$\Sigma _{t}^{\omega }$. We first treat two trivial cases:

\begin{itemize}
\item If $X_{t+1}(\omega )=\infty $ for some $\omega \in \Omega _{\ast }$
then the claim is not super-replicable at a finite cost hence the thesis
follows with $X_{0}=\sup_{Q\in \mathcal{M}_{f}}X_{0}(Q)=\infty $.

\item If $\Sigma _{t}^{\omega }\cap \Omega _{\ast }^{t+1}=\varnothing $ we
have two consequences: $\Sigma _{t}^{\omega }$ is an $\mathcal{M}_{f}$-polar
set, hence by assumption, $X_{t}(Q)=-\infty $ on $\Sigma _{t}^{\omega }$,
for any $Q\in \mathcal{M}_{f}$. Moreover, as explained after equation %
\eqref{Replicating}, $\Pi _{x_{d+1}}(R_{t+1,X_{t+1},D_{t+1}})=\mathbb{R}$ so
that $X_{t}(\omega )=-\infty $ and the desired equality follows.
\end{itemize}

From now on we therefore assume $X_{t+1}<\infty $ and $\Sigma _{t}^{\omega
}\cap \Omega _{\ast }^{t+1}\neq \varnothing $. Define, for any $y\in \mathbb{%
R}$, the set
\begin{equation*}
\Gamma _{y}:=co\left( conv\left\{ \left[ \Delta S_{t+1}(\widetilde{\omega }%
);y-X_{t+1}(\widetilde{\omega })\right] \mid \widetilde{\omega }\in \Sigma
_{t}^{\omega }\cap \Omega _{\ast }^{t+1}\right\} \right)
\end{equation*}%
We claim that
\begin{equation}
0\in int(\Gamma _{y})\Longrightarrow X_{t}>y  \label{interiorY}
\end{equation}%
Indeed from $0\in int(\Gamma _{y})$ there is no non zero $(H,h)\in \mathbb{R}%
^{d}\times \mathbb{R}$ , such that either $h(y-X_{t+1})+H\cdot \Delta
S_{t+1}\geq 0$ or $h(y-X_{t+1})+H\cdot \Delta S_{t+1}\leq 0$ on $\Sigma
_{t}^{\omega }\cap \Omega _{\ast }^{t+1}$. In particular there is no $H\in
\mathbb{R}^{d}$ such that
\begin{equation}  \label{yNotX}
y+H\cdot \Delta S_{t+1}\geq X_{t+1}\text{ on } \Sigma _{t}^{\omega }\cap
\Omega _{\ast }^{t+1}
\end{equation}
Recalling that, by definition, $X_{t}$ is the infimum of real numbers for
which \eqref{yNotX} is satisfied, we have $X_t\geq y$. Since, from Step 2, $%
X_{t}$, when finite, is actually a minimum, we have $X_t>y$ and %
\eqref{interiorY} follows.

\begin{enumerate}
\item[Premise:] As in Step 1, we may suppose, without loss of generality,
that if for some $H\in \mathbb{R}^{d}$, $H\cdot \Delta S_{t+1}=0$ on $\Sigma
_{t}^{\omega }\cap \Omega _{\ast }^{t+1}$ then $H=0$. In fact if this is not
the case we can reduce, with an analogous procedure, the number of assets
needed for super-replication on the level set .

We now distinguish two cases.

\item[\textbf{Case 1:}] Suppose there exist $(H,h,\alpha )\in \mathbb{R}%
^{d+2}$ with $(H,h,\alpha )\neq 0$ such that $h(y-X_{t+1})+H\cdot \Delta
S_{t+1}=\alpha $ on $\Sigma _{t}^{\omega }\cap \Omega _{\ast }^{t+1}$. We
claim that $h\neq 0$. Indeed, if $h=0$ then $\alpha \neq 0,$ since $H\cdot
\Delta S_{t+1}=0$ implies $(H,h,\alpha )=0$. However, $\alpha \neq 0$
implies $H\cdot \Delta S_{t+1}=\alpha $ on $\Sigma _{t}^{\omega }\cap \Omega
_{\ast }^{t+1}$ which would yield a trivial one point arbitrage on $\Omega
_{\ast }$, hence a contradiction. \newline
Since $h\neq 0$ we have $y-\frac{\alpha }{h}+\frac{H}{h}\cdot \Delta
S_{t+1}=X_{t+1}$ on $\Sigma _{t}^{\omega }\cap \Omega _{\ast }^{t+1}$: this
means that $X_{t}$ from Step 3 coincides with $y-\frac{\alpha }{h}$ and $%
X_{t+1}$ is replicable implementing the strategy $\bar{H}:=\frac{H}{h}$ in
the risky assets and $X_t=y-\frac{\alpha }{h}$ in the non-risky asset. If
now for some $Q\in \mathcal{M}_{f}$ such that $Q(\Sigma _{t}^{\omega })>0$,
we have the existence of $x\leq X_{t}$ and $H_{x}\in \mathbb{R}^{d}$ such
that $x+H_{x}\cdot \Delta S_{t+1}\geq X_{t+1}$ $Q$-a.s. then $x-X_{t}+(H_{x}-%
\bar{H})\Delta S_{t+1}\geq 0$ $Q$-a.s. hence, since $NA(Q)$ holds true, $%
x\geq X_{t}$. Therefore $X_{t}=X_{t}(Q)$ on $\Sigma _{t-1}^{\omega }$.

\item[\textbf{Case 2:}] If a triplet $(H,h,\alpha )\in \mathbb{R}^{d+2}$
such as in Case 1 does not exist then we define

\begin{equation*}
\bar{y}=\sup \left\{ y\in \mathbb{R}\mid \exists \,H\in \mathbb{R}%
^{d}:\;y+H\cdot \Delta S_{t+1}\leq X_{t+1}\text{ on }\Sigma _{t}^{\omega
}\cap \Omega _{\ast }^{t+1}\right\} .
\end{equation*}%
Obviously $\bar{y}<X_{t}$ otherwise we are back to Case 1. For every $%
0<\varepsilon <X_{t}-\bar{y}$ and for every $H\in \mathbb{R}^{d}$ neither $%
X_{t}-\varepsilon +H\Delta S_{t+1}\geq X_{t+1}$ nor $X_{t}-\varepsilon
+H\Delta S_{t+1}\leq X_{t+1}$ holds true on $\Sigma _{t}^{\omega }\cap
\Omega _{\ast }^{t+1}$. Moreover if there exists $h\in \mathbb{R}$ such that
$h(X_{t}-\varepsilon -X_{t+1})+H\Delta S_{t+1}\geq 0$ (or $%
h(X_{t}-\varepsilon -X_{t+1})+H\Delta S_{t+1}\leq 0$) on $\Sigma
_{t}^{\omega }\cap \Omega _{\ast }^{t+1}$ necessarily $h$ would be 0
(otherwise simply divide by $h$). In such a case $H\Delta S_{t+1}\geq 0$ (or
$H\Delta S_{t+1}\leq 0$) on $\Sigma _{t}^{\omega }\cap \Omega _{\ast }^{t+1}$
and by absence of one point arbitrage we get $H\Delta S_{t+1}=0$ and hence $%
H=0$. For this reason neither $h(X_{t}-\varepsilon -X_{t+1})+H\Delta
S_{t+1}\geq 0$ nor $h(X_{t}-\varepsilon -X_{t+1})+H\Delta S_{t+1}\leq 0$ for
any $(H,h)\in \mathbb{R}^{d+1}\setminus \{0\}$ so that $0\in int\Gamma
_{X_{t}-\varepsilon }$. \newline
Take $\{\omega _{i}\}_{i=1}^{k}\subset \Sigma _{t}^{\omega }\cap \Omega
_{\ast }$ (with $k\leq d$) such that $\{\left[ \Delta S_{t+1}(\omega
_{i});X_{t}-\varepsilon -X_{t+1}(\omega _{i})\right] \mid i=1,\ldots ,k\}$
are linearly independent and generates the same linear space in $\mathbb{R}%
^{d+1}$ as $\Gamma _{X_{t}-\varepsilon }$. By Proposition \ref{LemmaNOpolar}%
, and the convexity of the set of martingale measures, there exists $Q\in
\mathcal{M}_{f}$ such that $Q(\{\omega _{i}\})>0$ for any $i=1,\ldots ,k$.
For such a $Q$ we get
\begin{equation*}
\Gamma _{X_{t}-\varepsilon }=co\left( conv\{\left[ \Delta S_{t+1}(\widetilde{%
\omega });X_{t}-\varepsilon -X_{t+1}(\widetilde{\omega })\right] \mid
\widetilde{\omega }\in supp(Q)\cap \Sigma _{t}^{\omega }\}\right)
\end{equation*}%
so that, from $0\in int\Gamma _{X_{t}-\varepsilon }$, there exists no $%
H(Q)\in \mathbb{R}^{d}$ such that $X_{t}-\varepsilon +H(Q)\cdot \Delta
S_{t+1}\geq X_{t+1}$ $Q$-a.s. We can conclude that $X_{t}\geq \sup_{Q\in
\mathcal{M}_{f}}X_{t}(Q)\geq X_{t}-\varepsilon $. Letting $\varepsilon
\downarrow 0$ we get $\sup_{Q\in \mathcal{M}_{f}}X_{t}(Q)=X_{t}$ as desired.
\end{enumerate}

\textbf{Step 5:} finally we prove (\ref{intersection}). Notice that $%
\mathcal{C}\subseteq \bigcap_{Q\in \mathcal{M}_{f}}\mathcal{C}(Q)$. Moreover
if $g\in \bigcap_{Q\in \mathcal{M}_{f}}\mathcal{C}(Q)$ then \eqref{splitQ}
holds with $X_{0}(Q)\leq 0$ for every $Q\in \mathcal{M}_{f}$. Therefore also
in Equation \eqref{split} we have $X_{0}=\sup_{Q\in \mathcal{M}%
_{f}}X_{0}(Q)\leq 0$ and $g\leq \sum_{t=1}^{T}H_{t}\cdot \Delta S_{t}$ on $%
\Omega _{\ast }$ i.e. $g\in \mathcal{C}$.

\bigskip

\begin{remark}
Note that the proof of Proposition \ref{supsup} relies only on the fact that
$\Omega_{\ast}$ is an analyitc set and that $(\Omega_{\ast})^C$ is the
maximal polar set for the class of finite support martingale measure. Given $%
\hat{\Omega}\subseteq \Omega$ an analytic subset of $(\Omega,\mathcal{F})$,
from Proposition \ref{propAn} it also follows that
\begin{equation*}
\hat{\mathcal{C}} =\bigcap_{\{Q\in \mathcal{M}_{f}\mid Q(\hat{\Omega})=1\}}%
\mathcal{C}(Q)
\end{equation*}
where $\hat{\mathcal{C}}:= \{f\in \mathcal{L}(\Omega ,\mathcal{F})\mid f\leq
k\text{ on }\hat{\Omega}_{\ast }\text{ for some }k\in \mathcal{K}\}$ and $%
\hat{\Omega}_{\ast }$ as in \eqref{omegaHat*}.
\end{remark}

\subsection{Proof of Theorem \protect\ref{superHO}\label{secOption}}

Recall that $\pi _{\Phi }$ is defined in (\ref{piPhi}) and $\mathcal{M}
_{\Phi }$ in (\ref{Mphi}). Set
\begin{equation*}
\widetilde{\pi }_{\Phi }(g):=\inf \left\{ x\in \mathbb{R}\mid \exists H\in
\mathcal{H}\text{ such that }x+(H\cdot S)_{T}(\omega )\geq g(\omega )\
\forall \omega \in \Omega _{\Phi }\right\} .
\end{equation*}

\begin{lemma}
\label{lemmag}Let $g:\Omega \mapsto \mathbb{R}$ and $\phi ^{j}:\Omega
\mapsto \mathbb{R}$, $j=1,...,k,$ be $\mathcal{F}$-measurable random
variables. Then
\begin{equation*}
\pi _{\Phi }(g)=\inf_{h\in \mathbb{R}^{k}}\widetilde{\pi }_{\Phi }(g-h\Phi ).
\end{equation*}
\end{lemma}

\begin{proof}
For every $h\in \mathbb{R}^{k}$ we have $\pi _{\Phi }(g)\leq \widetilde{\pi }
_{\Phi }(g-h\Phi )$ so that $\pi _{\Phi }(g)\leq \inf_{h\in \mathbb{R}^{k}}
\widetilde{\pi }_{\Phi }(g-h\Phi )$. By contradiction assume $\pi _{\Phi
}(g)<\inf_{h\in \mathbb{R}^{k}}\widetilde{\pi }_{\Phi }(g-h\Phi )$, then
there exist $(\bar{x},\bar{h},\bar{H})\in (\mathbb{R},\mathbb{R}^{k},
\mathcal{H})$ such that
\begin{eqnarray*}
&&\bar{x}<\inf_{h\in \mathbb{R}^{k}}\widetilde{\pi }_{\Phi }(g-h\Phi )\quad
\text{ and } \\
&&\bar{x}+(\bar{H}\cdot S)_{T}(\omega )+\bar{h}\Phi (\omega )\geq g(\omega )
\text{ for all }\omega \in \Omega _{\Phi }
\end{eqnarray*}
Clearly we have a contradiction since
\begin{equation*}
\bar{x}<\widetilde{\pi }_{\Phi }(g-\bar{h}\Phi )=\inf \left\{ x\in \mathbb{R}
\mid \exists H\in \mathcal{H}\text{ s. t. }x+(H\cdot S)_{T}(\omega )\geq
g(\omega )-\bar{h}\Phi (\omega )\ \forall \omega \in \Omega _{\Phi }\right\}
\leq \bar{x}.
\end{equation*}
\end{proof}

\begin{proof}[Proof of Theorem \protect\ref{superHO}]
Since also $\Omega _{\Phi }$ is analytic (Proposition \ref{propAn}), by
comparing the definition of $\Omega _{\Phi }$ in (\ref{omegaphi}) with (\ref%
{omega*f}), we may repeat step by step the same arguments used in the proof
of Theorem \ref{superH} and Proposition \ref{supsup} replacing $\Omega
_{\ast }$ with $\Omega _{\Phi }$. We then conclude that $\widetilde{\pi }
_{\Phi }(g)=\sup_{\{Q\in \mathcal{M}_{f}\mid supp(Q)\subseteq
\Omega_{\phi}\}}E_{Q}[g]$ for any $\mathcal{F}$-measurable random variable $g
$. From the hypothesis we also have $\widetilde{\pi } _{\Phi }(g)=\sup_{Q\in
\mathcal{M}_{\Phi }}E_{Q}[g]$. Since $E_{Q}[h\Phi ]=0$ for all $Q\in
\mathcal{M}_{\Phi }$ and $h\in \mathbb{R}^{k},$ for the $\mathcal{F}$%
-measurable random variable $g-h\Phi $ we have
\begin{equation*}
\widetilde{\pi }_{\Phi }(g-h\Phi )=\sup_{Q\in \mathcal{M}_{\Phi
}}E_{Q}[g-h\Phi ]=\sup_{Q\in \mathcal{M}_{\Phi }}E_{Q}[g],\text{ }\forall
h\in \mathbb{R}^{k}\text{.}
\end{equation*}
The Lemma \ref{lemmag}\ then implies: $\pi _{\Phi }(g)=\inf_{h\in \mathbb{R}
^{k}}\widetilde{\pi }_{\Phi }(g-h\Phi )=\sup_{Q\in \mathcal{M}_{\Phi
}}E_{Q}[g]$.
\end{proof}


\begin{thebibliography}{HL0ST16}
\bibitem[AB16]{AB13} {Acciaio B., Beiglb\"{o}ck M., Penkner F.,
Schachermayer W., A model-free version of the fundamental theorem of asset
pricing and the super-replication theorem, \textit{Math. Fin.}, 26(2),
233-251, 2016.}

\bibitem[AB06]{Aliprantis} {Aliprantis C. D., Border K. C., \textit{Infinite
Dimensional Analysis}, Springer, Berlin, 2006.}

\bibitem[BHLP13]{BHLP13} {\ Beiglb\"{o}ck M., Henry-Labord\`ere P., Penkner
F., Model-independent bounds for option prices: a mass transport approach,
\textit{Fin. Stoch.} 17(3), 477-501, 2013.}

\bibitem[BS78]{BS78} {Bertsekas, D.P., Shreve S.E., \textit{Stochastic
Optimal Control. The discrete time case}, Academic Press, New York, 1978.}


\bibitem[BN15]{BN13} {Bouchard B., Nutz M., Arbitrage and Duality in
Nondominated Discrete-Time Models, \textit{Ann. Appl. Prob.}, 25(2),
823-859, 2015.}

\bibitem[BHR01]{BHR01} {Brown H.M., Hobson D.G., Rogers L.C.G., Robust
hedging of barrier options, \textit{Math. Fin.}, 11, 285-314, 2001.}

\bibitem[BFM16]{BFM14} {Burzoni M., Frittelli M.,Maggis M., Universal
Arbitrage Aggregator in discrete time Markets under Uncertainty, \textit{%
Fin. Stoch.}, 20(1), 1-50, 2016.}

\bibitem[CO11]{CO11} {Cox A.M.G., Ob\l oj J., Robust pricing and hedging of
double no-touch options, \textit{Fin. Stoch.}, 15(3),573-605, 2011.}

\bibitem[DS94]{DS94} {Delbaen F., Schachermayer W., A general version of the
fundamental theorem of asset pricing, \textit{Math. Ann.}, 300,463-520, 1994.%
}

\bibitem[DM82]{DM82} {Dellacherie C., Meyer P., Probabilities and Potential
B, North-Holland, Amsterdam New York 1982}

\bibitem[DS13]{DS13} {Dolinsky Y., Soner H. M., Martingale optimal transport
and robust hedging in continuous time, \textit{Probab. Theory Related Fields}%
, 160(1-2), 391-427, 2013.}

\bibitem[DS14]{DS14} {Dolinsky Y., Soner H. M., Robust hedging with
proportional transaction costs, \textit{Fin. Stoch.}, 18 (2), 327-347, 2014.}

\bibitem[DS15]{DS14b} {Dolinsky Y., Soner H. M., Martingale optimal
transport in the Skorokhod space, \textit{Stoch. Proc. Appl.}, 125(10),
3893-3931, 2015.}

\bibitem[KQ95]{ekq} El Karoui N. and M.C. Quenez, Dynamic programming and
pricing of contingent claims in an incomplete market, \textit{SIAM Journal
of Contr. and Opt.} \textbf{33}, 29-66, 1995.

\bibitem[GHLT14]{GHLT14} {Galichon A., Henry-Labord\`ere P., Touzi N., A
stochastic control approach to no-arbitrage bounds given marginals, with an
application to lookback options, \textit{Ann. Appl. Prob}, 24 (1), 312-336,
2014.}

\bibitem[HL0ST16]{HLOST15} {\ Henry-Labord\`{e}re P., Ob\l oj J., Spoida P.,
Touzi N., The maximum maximum of a martingale with given n marginals,
\textit{Ann. Appl. Prob.}, 26(1), 1-44, 2016.}

\bibitem[Ho98]{Ho98} {Hobson D.G., Robust hedging of the lookback option,
\textit{Fin. Stoch.}, 2(4), 329-347, 1998.}

\bibitem[Ho11]{Ho11} {Hobson D.G., The Skorokhod embedding problem and
model-independent bounds for option prices, \textit{Paris-Princeton Lectures
on Math. Fin. 2010}, Volume 2003 of \textit{Lecture Notes in Math,},
267-318, Springer-Berlin 2011.}

\bibitem[KS01]{KS01} Kabanov, Y. M. and Stricker, C., A teacher note on
no-arbitrage criteria. \textit{Sem. Probab. 35, Lecture Notes in Math.}
1755, Springer-Verlag, 149--2, 2001.

\bibitem[Ka97]{Ka} Karatzas I., Lectures on the Mathematics of Finance, CRM
Monograph Series 8, 148 pp. \textit{American Mathematical Society}, 1997.

\bibitem[OH15]{OH15} {Ob\l oj J., Hou Z., On robust pricing-hedging duality
in continuous time, preprint, 2015.}

\bibitem[Ri15]{Riedel} {Riedel F., Financial economics without probabilistic
prior assumptions, \textit{Dec. Econ. Fin.}, 38 (1), 75-91, 2015.}

\bibitem[RW98]{R} {Rockafellar T, Wets R., Variational Analysis, Springer
1998}

\bibitem[TT13]{TT13} {Tan X., Touzi N., Optimal Transportation under
Controlled Stochastic Dynamics, \textit{Ann. Prob.}, 41(5), 3201-3240, 2013.}
\end{thebibliography}
\end{document}